\documentclass[12pt]{article}
\usepackage{hyperref}
\usepackage{float}
\usepackage{amsmath}
\usepackage{graphicx,psfrag,epsf}
\usepackage{enumerate}
\usepackage{natbib}
\usepackage{multirow}
\usepackage[ruled,vlined]{algorithm2e}
\usepackage{array}
\usepackage{xcolor}
\allowdisplaybreaks

\usepackage{graphicx}
\usepackage[utf8 ]{inputenc}
\usepackage[english]{babel}
\usepackage{amsthm}
\usepackage[bottom]{footmisc}
\usepackage{longtable}
\usepackage{mathrsfs}
\usepackage{extarrows}
\usepackage{bigints}
\usepackage{subcaption}
\usepackage{setspace}

\usepackage[top=0.85in, bottom=0.85in, left=0.85in, right=0.85in]{geometry}
\usepackage{amsmath}
\usepackage{amsfonts}

\usepackage{rotating}
\usepackage{graphicx} 
\usepackage{threeparttable, tablefootnote}
\usepackage{hyperref}

\newtheorem{theorem}{Theorem}[section]

\newtheorem{proposition}[theorem]{Proposition}

\theoremstyle{definition}

\newtheorem{definition}{Definition}[section]

\DeclareTextFontCommand{\textmyfont}{\footnotesize}

\begin{document}


\doublespacing
	
	\title{\bf A multivariate heavy-tailed integer-valued GARCH process with EM algorithm-based inference}

	\author{Yuhyeong Jang$^\star$\footnote{Email: \href{mailto:yuhyeongj@smu.edu}{yuhyeongj@smu.edu}}\,\,\, \, Raanju R. Sundararajan$^\star$\footnote{Email: \href{mailto:rsundararajan@smu.edu}{rsundararajan@smu.edu}}\,\,\, \, Wagner Barreto-Souza\footnote{Email: \href{mailto:wagner.barreto-souza@ucd.ie}{wagner.barreto-souza@ucd.ie}} \\
	$^\star${\it \small Department of Statistical Science, Southern Methodist University,  Dallas, Texas, USA} \\
	$^\ddag${\it \small School of Mathematics and Statistics, University College Dublin, Dublin, Ireland}}

\vspace{0.5cm}

	\date{}

	\maketitle
	
	\begin{abstract}
		A new multivariate integer-valued Generalized AutoRegressive Conditional Heteroscedastic process based on a multivariate Poisson generalized inverse Gaussian distribution is proposed.  The estimation of parameters of the proposed multivariate heavy-tailed count time series model via maximum likelihood method is challenging since the likelihood function involves a Bessel function that depends on the multivariate counts and its dimension. As a consequence, numerical instability is often experienced in optimization procedures. To overcome this computational problem, two feasible variants of the Expectation-Maximization (EM) algorithm are proposed for estimating parameters of our model under low and high-dimensional settings. These EM algorithm variants provide computational benefits and help avoid the difficult direct optimization of the likelihood function from the proposed model. 
		Our model and proposed estimation procedures can handle multiple features such as modeling of multivariate counts, heavy-taildness, overdispersion, accommodation of outliers, allowances for both positive and negative autocorrelations, estimation of cross/contemporaneous-correlation, and the efficient estimation of parameters from both statistical and computational points of view. Extensive Monte Carlo simulation studies are presented to assess the performance of the proposed EM algorithms. An application to modeling bivariate count time series data on cannabis possession-related offenses in Australia is discussed.\\
		
		\noindent {\bf MOS subject Classification}. Primary:  62M09; 62M10.\\
		
		\end{abstract}

	\section{Introduction}
	\label{sec:intro}
	
\noindent In recent years, modeling multivariate count time series data has been receiving increasingly more attention from social scientists as well as researchers from other disciplines. Multivariate count time series data can be commonly found in many fields such as economics, sociology, public health, finance and environmental science. Examples include weekly number of deaths from a disease observed across multiple geographical regions \citep{Paul2008}, high-frequency bid and ask volumes of a stock \citep{Pedeli2013}, and annual numbers of major hurricanes in the North Atlantic and North Pacific Basins \citep{Livsey2018}, to name a few. Many of these count time series examples exhibit heavy-tailed behavior and this poses a non-trivial modeling challenge.

Several methods exist in the literature for modeling univariate count time series data. Among them is the 
INteger-valued Generalized AutoRegressive Conditional Heteroscedastic (INGARCH) approach; see \cite{ferland06}, \cite{fokianos2009}, \cite{fokianos2011} for some examples.  \cite{wagner19} is a recent work that considers a class of INGARCH models using mixed Poisson (MP) distributions. One of their  proposed models, the Poisson inverse-Gaussian INGARCH process, viewed as an alternative to the negative binomial model,  is capable of handling heavy-tailed behavior in count time series, and is also robust to the presence of outliers.  Other existing models to handle univariate heavy-tailed time series of counts are due to \cite{bar2019}, \cite{gor2020}, and \cite{Qian2020}.

Following a number of studies on univariate models for time series of counts, some researchers have made efforts to generalize univariate models to bivariate and multivariate cases. As examples, \cite{pedeli2011} proposed a bivariate INteger-valued AutoRegressive (INAR) model and \cite{fokianos2020} introduce a multivariate count autoregression by combining copulas and INGARCH approaches. Readers are  referred to  \cite{davis2021} and \cite{fokianos2022} for a more detailed review. Bivariate INGARCH models have been proposed and explored by \cite{cuizhu2018}, \cite{leeetal2018}, \cite{cuietal2020}, and \cite{pianetal2023}.

Very limited research is available to address the problem of modeling multivariate count time series that exhibit heavy-tailed behavior. Such a feature is commonly seen in many real-data situations such as numbers of insurance claims for different types of properties \citep{chen2023}, and high-frequency trading volume in the financial market \citep{Qian2020}, to name a few.

	The primary aim of this paper is to introduce a novel multivariate time series model for handling heavy-tailed counts. We consider a Poisson generalized inverse-Gaussian distribution for the count vector given the past thereby allowing for the presence of heavy-taild behavior and outliers. Moreover, the model construction follows a log-linear INGARCH structure, which enables us to account for both negative and positive autocorrelations. We refer to our model as multivariate Poisson generalized inverse-Gaussian  INGARCH (MPGIG-INGARCH) process. A challenging point is that the resulting likelihood function of our multivariate count time series model depends on the Bessel function, posing difficulties and instability in estimating parameters via the conditional maximum likelihood method. Furthermore, as the dimension, denoted by $p$, increases, the number of parameters (equal to $2p^2+p+2$) grows fast. Hence, a direct optimization of the likelihood is computationally problematic, especially when the dimension grows. We properly address this issue by proposing effective Expectation-Maximization algorithms.
	
	 The main contributions of this paper are highlighted in what follows. 

\begin{itemize}

\item[(i)] The MPGIG-INGARCH process and the proposed estimation procedures can handle multiple features such as modeling of multivariate counts, heavy-taildness, overdispersion, accommodation of outliers, allowances for both positive and negative autocorrelations, estimation of cross/contemporaneous-correlation, and the efficient estimation of parameters from both statistical and computational points of view. This is the first model capable of simultaneously fulfilling all these functionalities.

\item[(ii)] The stochastic representation of our model enables the use of an Expectation-Maximization (EM) algorithm for efficiently estimating the parameters. This helps avoid computational challenges that come with the direct optimization of the likelihood function of the MPGIG-INGARCH model.

\item[(iii)] The development of a new Generalized Monte Carlo Expectation-Maximization (GMCEM) algorithm  that is computationally tractable and stable. For the larger dimensional cases, we also propose a new hybrid version of the GMCEM algorithm that combines the EM approach and quasi-maximum likelihood estimation (QMLE). In this hybrid approach, the QMLE technique is employed to estimate the parameters controlling temporal dependence, while the GMCEM algorithm is used to estimate the parameters that control overdispersion or degree of heavy-tailedness. The performance and computational burden of the GMCEM algorithm is inspected in an extensive simulation study given in Section \ref{sec:simulations}. To the best of our knowledge, it is the first time that such a hybrid EM algorithm variant has been suggested, in particular, for dealing with count time series. We believe that this proposed algorithm can be used more generally for problems where a subset of parameters can be well-estimated via quasi-likelihood and the remaining ones can be estimated via an EM algorithm.

\item[(iv)] Our model is uniquely positioned to handle heavy-tailed behavior found in many real-data examples. The motivating application comes from monthly counts\footnote{Data source: \url{https://www.bocsar.nsw.gov.au/Pages/bocsar_datasets/Offence.aspx}} of cannabis possession-related offenses recorded in the state of New South Wales, Australia. In particular, we consider a bivariate time series of counts recorded in the Mid North Coast (MNC) and Greater Newcastle (GNC) regions of New South Wales, Australia. This count time series data is seen to exhibit heavy-tailed behavior and our analysis shows that the proposed method is uniquely positioned in effectively describing this dataset.

\end{itemize}
	
A related work to ours is due to \cite{chen2023}, where a first-order multivariate count model based on the INAR approach with innovations following a multivariate Poisson generalized inverse-Gaussian distribution is introduced. The estimation of parameters of their model is performed via the method of  maximum likelihood. It is worth mentioning that likelihood functions of INAR models are generally quite cumbersome, even in the univariate case, and this issue is exacerbated in the multivariate context, which poses challenges regarding estimation and prediction \citep{fokianos2020}. Besides this model not having a computationally feasible estimation technique, especially for higher dimensions, it does not handle all the features of our model mentioned in (i) above, for instance, allowing for negative autocorrelations. Moreover, the parameter related to the tail is fixed by  \cite{chen2023} in their model to avoid non-identifiability problems. In our case, we do not face this issue and such a parameter is efficiently estimated.
	
	The rest of this paper is organized as follows. Section \ref{sec:model} introduces the MPGIG-INGARCH model along with other necessary technical details. In Section \ref{sec:estimation}, we describe the proposed GMCEM estimation algorithm, and also the hybrid version H-GMCEM algorithm that will be used in the larger dimensional cases. The results of Monte Carlo simulations are discussed in Section \ref{sec:simulations}. The application to modeling counts of cannabis-related offenses in two regions in Australia is presented in Section \ref{sec:application}. Concluding remarks and future research are provided in Section \ref{sec:conclusion}.

\section{MPGIG$_p$-INGARCH model}
\label{sec:model}

In this section we describe our multivariate INGARCH model that is designed to handle heavy-tailed behavior in count time series data. We begin by introducing some notations and definitions that are needed followed by a description of our proposed model. 

We say that random variable $Z$ follows a generalized inverse-Gaussian (GIG) distribution if its probability density function assumes the form
	\begin{eqnarray}\label{densityGIG}
		g(z)=\dfrac{(a/b)^{\alpha/2}}{2\mathcal K_\alpha(\sqrt{ab})}z^{\alpha-1}\exp\left\{-\dfrac{1}{2}\left(az+bz^{-1}\right)\right\},\quad z>0,	
	\end{eqnarray}	
	where $a>0$, $b>0$, and $\alpha\in\mathbb R$ are parameters, and $\mathcal K_{\alpha}(z)=\dfrac{1}{2}\displaystyle\int_0^\infty u^{\alpha-1}\exp\{-z(u+u^{-1})/2\}du$ is the second-kind modified Bessel function. We denote this three-parameter distribution by $Z\sim\mbox{GIG}(a,b,\alpha)$. The bi-parameter case when $a=b=\phi$, which will be considered in this paper, is denoted by $Z\sim\mbox{GIG}(\phi,\alpha)$.  
	The multivariate Poisson generalized inverse-Gaussian distribution \citep{steinetal1987}, here denoted as MPGIG$_p$, is defined by assuming that $Y_1,\ldots,Y_p$ are conditionally independent random variables given $Z\sim\mbox{GIG}(\phi,\alpha)$, and that $Y_i|Z\sim\mbox{Poisson}(\lambda_iZ)$, for $i=1,\ldots,p$. The joint probability mass function of $Y_1,\ldots,Y_p$ is given by
	\begin{eqnarray*}
		P(Y_1=y_1,\ldots,Y_p=y_p)=\dfrac{\mathcal K_{\sum_{i=1}^py_i+\alpha}\left(\sqrt{\phi\left(2\sum_{i=1}^p\lambda_i+\phi\right)}\right)}{\mathcal K_\alpha(\phi)}\left(\prod_{i=1}^p\dfrac{\lambda_i^{y_i}}{y_i!}\right)\left(\dfrac{\phi}{2\sum_{i=1}^p\lambda_i+\phi}\right)^{\frac{\sum_{i=1}^py_i+\alpha}{2}},		
	\end{eqnarray*}	
	for $y_1,\ldots,y_p\in\mathbb N_0\equiv\{0,1,2,\ldots\}$. We denote ${\bf Y}=(Y_1,\ldots,Y_p)^\top\sim\mbox{MPGIG}_p(\boldsymbol\lambda,\phi,\alpha)$, with $\boldsymbol\lambda=(\lambda_1,\ldots,\lambda_p)^\top$. Now define $\mathcal R_{\alpha,k}(\phi)=\dfrac{\mathcal K_{k+\alpha}(\phi)}{\mathcal K_\alpha(\phi)}$,  for $k\in\mathbb N$, with $\mathcal R_{\alpha}(\phi)\equiv\mathcal R_{\alpha,1}(\phi)$. The first two cumulants and covariance of the MPGIG$_p$ distribution can be expressed by $E(Y_i)=\lambda_i\mathcal R_{\alpha}(\phi)$, $\mbox{Var}(Y_i)=\lambda_i\mathcal R_{\alpha}(\phi)+\lambda_i^2\{\mathcal R_{\alpha,2}(\phi)-\mathcal R^2_{\alpha}(\phi)\}$, for $i=1,\ldots,p$, and $\mbox{cov}(Y_i,Y_j)=\lambda_i\lambda_j\{\mathcal R_{\alpha,2}(\phi)-\mathcal R^2_{\alpha}(\phi)\}$, for $i\neq j$.
	
 We use the notation $\{{\bf Y}_t\}_{t\geq1}=\{ ( Y_{1t},\ldots,Y_{pt})\}_{t\geq1}$ to denote a $p$-variate time series of counts in what follows for defining our multivariate INGARCH model.
	
	\begin{definition} ($\mbox{MPGIG}_p\mbox{-INGARCH}$ process)
		\label{def:MPGIGINGARCH}
		We say that $\{{\bf Y}_t\}_{t\geq1}$ follows a MPGIG$_p$-INGARCH process if ${\bf Y}_t|\mathcal F_{t-1}\sim\mbox{MPGIG}_p(\boldsymbol\lambda_t,\phi,\alpha)$, with $\mathcal F_{t-1}=\sigma({\bf Y}_{t-1},\ldots, {\bf Y}_1,\boldsymbol\lambda_1)$, and $\boldsymbol\nu_t\equiv \log\boldsymbol\lambda_t$ defined componentwise and satisfying the dynamics
		\begin{eqnarray}
			\boldsymbol\nu_t={\bf d}+{\bf A}\boldsymbol\nu_{t-1}+{\bf B}\log({\bf Y}_{t-1}+{\bf 1}_p),	\quad t\geq1,
			\label{eqn::meanprocess}
		\end{eqnarray}	
		where ${\bf d}$ is a $p$ dimensional vector with real-valued components, ${\bf A}$ and ${\bf B}$ are $p\times p$ real-valued matrices, and ${\bf 1}_p$ is a $p$-dimensional vector of ones.
	\end{definition}

As an illustration, in Figure \ref{fig:simuldata}, we present plots of the simulated trajectories of a 4-dimensional MPGIG$_4$-INGARCH process with the following specifications
$$\phi=0.5,\text{ }\alpha=1.5,\text{ }
\mathbf{\rm d} = \begin{bmatrix}
	0.5 \\
	0.5\\
	1\\
	0.5
\end{bmatrix},\text{ } 
\mathbf{A} = \begin{bmatrix}
	0.35 &-0.2 &0&0\\
	0 &-0.3 &0&0\\
	0 &0 &0.4&0\\
	0 &0 &0.2&-0.3
\end{bmatrix},\text{ } 
\mathbf{B}=\begin{bmatrix}
	-0.3 &0.2 &0.0&0\\
	0 &0.3 &0&0\\
	0 &0 &-0.3&0\\
	0 &0&-0.25&0.4
\end{bmatrix}.$$

\begin{figure}[H]
  \centering
  \includegraphics[width=17cm]{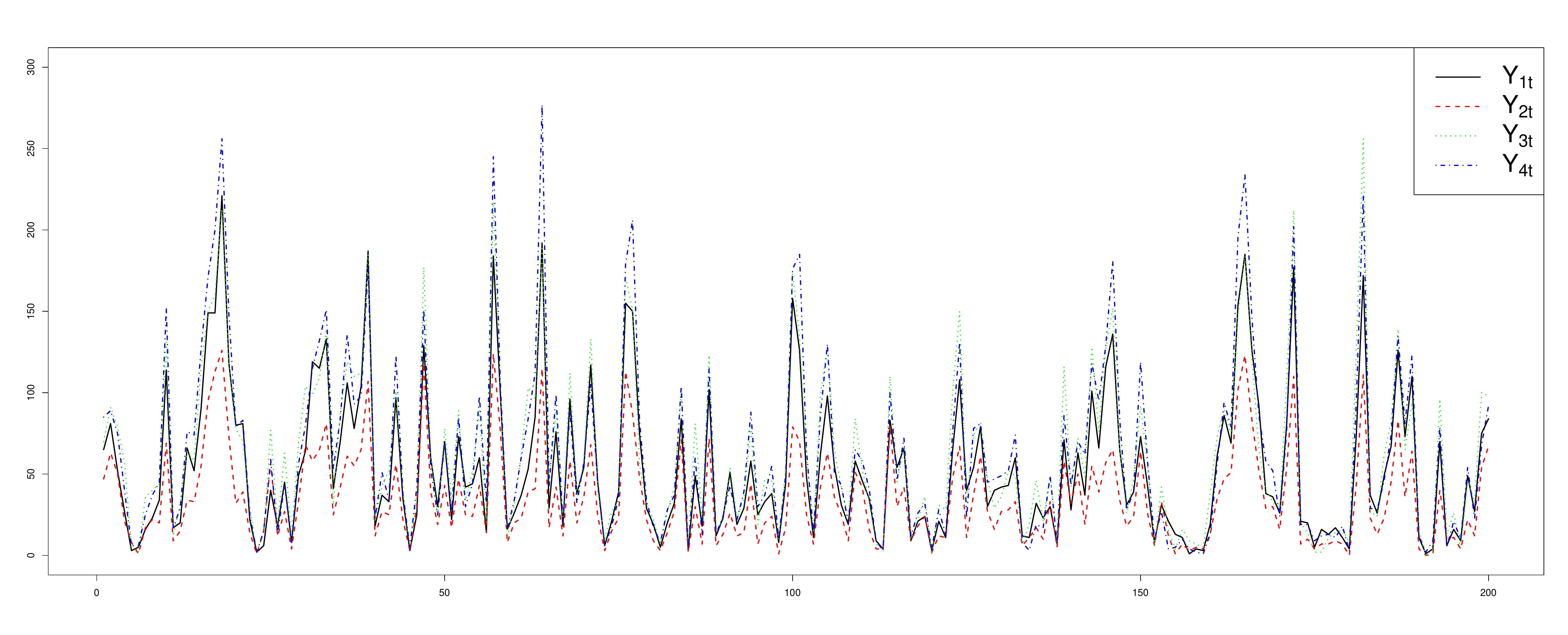}
  \caption{Simulated trajectories of a 4-dimensional MPGIG$_4$-INGARCH process with $T=200$.}
  \label{fig:simuldata}
\end{figure}

It is worth mentioning that the MPGIG$_p$-INGARCH process can accommodate both positive and negative autocorrelations due to the log-linear form assumed in (\ref{eqn::meanprocess}). The log-linear link function was previously proposed by \cite{fokianos2011} and \cite{fokianos2020} for univariate and multivariate count time series cases, respectively. In the latter work, the authors assume that the conditional distribution of each component of $\textbf{Y}_t$ follows a univariate Poisson distribution. In contrast, under our setup, the conditional distribution is MPGIG$_p$, which can handle heavy-taildness.

Going forward, we will assume that $\boldsymbol{\lambda}_{t}$ in the MPGIG$_p$-INGARCH model above is defined with the log-link function, unless mentioned otherwise.

The model parameters can be estimated by conditional maximum likelihood method. Denote the parameter vector as  $\boldsymbol{\Theta}=\left(\phi,\alpha,\boldsymbol{\theta}\right)^\top$, where $\boldsymbol{\theta}=\left(\textbf{d},\mbox{vec}(\textbf{A}),\mbox{vec}(\textbf{B})\right)^\top$.  The conditional log-likelihood function assumes the form
\begin{eqnarray}\label{eqn:pmf}
\ell(\boldsymbol{\Theta})&\equiv&\sum_{t=2}^T\log P({\bf Y}_t={\bf y}_t|\mathcal F_{t-1})\nonumber\\
&\propto&(T-1)\log\mathcal K_\alpha(\phi)+\sum_{t=2}^T\sum_{i=1}^py_{it}\log\lambda_{it}+\dfrac{\log\phi}{2}\left[\sum_{t=2}^T\sum_{i=1}^py_{it}+\alpha(T-1)\right] \nonumber \\
&&-\dfrac{1}{2}\sum_{t=2}^T\left(\sum_{i=1}^py_{it}+\alpha\right)\log\left(2\sum_{i=1}^p\lambda_{it}+\phi\right)\nonumber\\
&&+\sum_{t=2}^T\sum_{i=1}^p\log\left\{\mathcal K_{\sum_{i=1}^p y_{it}+\alpha}\left(\sqrt{\phi\left(2\sum_{i=1}^p\lambda_{it}+\phi\right)}\right)\right\}.
\end{eqnarray}	

The conditional maximum likelihood estimator (CMLE) of $\boldsymbol{\Theta}$ is given by $\widehat{\boldsymbol{\Theta}}=\mbox{argmax}_{\boldsymbol{\Theta}}\ell(\boldsymbol{\Theta})$.
Direct optimization of the conditional likelihood function (\ref{eqn:pmf}) is not straightforward due to its complicated form, particularly due to the presence of modified Bessel functions. Additionally, optimization algorithms are likely to become computationally heavier and more unstable as the number of parameters increases. To overcome these computational difficulties, EM algorithms will be proposed in the following section to perform inference for our multivariate count time series model.

	\section{EM algorithm inference}
	\label{sec:estimation}
Here we describe estimation methods for the MPGIG$_p$-INGARCH model from Section \ref{sec:model}. The direct optimization of the observed conditional log-likelihood from the MPGIG$_p$-INGARCH model tends to be unstable and computationally expensive due to the presence of a large number of parameters and its complicated form. To overcome this problem, two variants of the EM-algorithm are outlined, and these are computationally feasible techniques meant to handle parameter estimation in the low and high-dimensional settings.

Let $\{ {\bf Y}_t \}_{t=1}^{T} = \{ ( Y_{1t} , \ldots , Y_{pt} ) \}_{t=1}^{T} $ denote a random sample of size $T$ from the MPGIG$_p$-INGARCH model, and let ${\bf y}_t$ be a realization of ${\bf Y}_t$ at time $t$. As before, $\boldsymbol{\Theta}$ is the parameter vector. The complete data is composed by $\{{({\bf Y}_t,Z_t})\}_{t=1}^T$, where $Z_1,\ldots,Z_T$ are iid latent $\mbox{GIG}(\phi,\alpha)$ random variables obtained from the stochastic representation of a multivariate Poisson GIG distribution discussed at the beginning of Section 2.  Then, the conditional complete data likelihood can be written as	
\begin{align} \label{eq:complete_data_likelihood}
\mathcal{L}_C (\boldsymbol{\Theta}\mid\boldsymbol{y},\boldsymbol{z}) 
&=\prod\limits_{t=2}^{T} p(\boldsymbol{y}_t,z_t\mid \mathcal{F}_{t-1})\\\notag
&=\prod\limits_{t=2}^{T}\big\{ g\left(z_t;\phi,\alpha\right)\cdot p\left(\boldsymbol{y}_t\mid \mathcal{F}_{t-1},z_t\right)\big\}\\ \notag
&=\prod\limits_{t=2}^{T}\left\{   \frac{z_t^{\alpha-1}}{2\mathcal{K}_\alpha (\phi)}\cdot \exp\left(-\frac{\phi}{2}(z_t+z_t^{-1})\right)    \prod\limits_{i=1}^{p} \cfrac{e^{-z_{t}\lambda_{it}}\left(z_t\lambda_{it}\right)^{y_{it}} }{y_{it}!}\right\}, \notag
\end{align}
where $p(\boldsymbol{y}_t,z_t\mid \mathcal{F}_{t-1})$ denotes the joint density function of $({\bf Y}_t,Z_t)$ given $\mathcal F_{t-1}$, for $t\geq2$. Next, by taking the logarithm of $\mathcal{L}_C$ from \eqref{eq:complete_data_likelihood}, we have the complete data log-likelihood, denoted by $\ell_{C}(\boldsymbol{\Theta})$, written as
\begin{align} \label{eq:log_complete_data_likelihood}
\ell_{C}\left(\boldsymbol{\Theta}\right) &=\sum\limits_{t=2}^{T} \ell_t(\boldsymbol\Theta)= \sum\limits_{t=2}^{T} \log p(\boldsymbol{y}_t,z_t\mid \mathcal{F}_{t-1}) \\ \notag
&=\sum\limits_{t=2}^T\log\left\{ \frac{z_t^{\alpha-1}}{2\mathcal{K}_\alpha (\phi)}\cdot \exp\left(-\frac{\phi}{2}(z_t+z_t^{-1})\right) \right\} + \sum\limits_{t=2}^{T}\sum\limits_{i=1}^{p} \log \left\{\cfrac{e^{-z_{t}\lambda_{it}}\left(z_t\lambda_{it}\right)^{y_{it}} }{y_{it}!} \right\}\\ \notag
&=\underbrace{\alpha\sum\limits_{t=2}^T  \operatorname{log}z_t-(T-1)\cdot\operatorname{log}\mathcal{K}_{\alpha}(\phi)-\cfrac{\phi}{2}\sum\limits_{t=2}^T(z_t+z_t^{-1})}_{=\ell_{C_1}(\phi,\alpha)}\vspace{15pt}   \\ \notag 
\vphantom{\frac{1}{2}} \\[-10pt]
&\hspace{15pt}+\underbrace{\sum\limits_{t=2}^T \sum\limits_{i=1}^p y_{it}\operatorname{log}\left(\lambda_{it}\right)-  \sum\limits_{t=2}^T z_t \sum\limits_{i=1}^p\lambda_{it} }_{=\ell_{C_2}(\boldsymbol{\theta})} \; + \;  \text{const.}\\ \notag
&\propto \ell_{C_1}(\phi,\alpha) + \ell_{C_2}(\boldsymbol{\theta}),\notag
\end{align}

Observe that $\ell_C({\boldsymbol{\Theta}})$ is decomposed as a sum of a function of $(\phi,\alpha)$, plus a function of $\boldsymbol{\theta}$, which makes the optimization procedure easier. We now discuss the E and M steps associated with the EM algorithm.

	\subsubsection*{E-step}
	
	A type-EM algorithm can be derived by replacing the $t$-th term from (\ref{eq:log_complete_data_likelihood}) $\ell_t(\boldsymbol\Theta)$ by a conditional expectation for all $t$. More specifically, the $Q$-function to be optimized is given by
	\begin{eqnarray}\label{Qfunction}
	Q(\boldsymbol\Theta;\boldsymbol\Theta^{(r)})\equiv 	\sum_{t=2}^T E_{\mathcal F_{t-1}}(\ell_t(\boldsymbol\Theta)|{\bf Y}_t;\boldsymbol\Theta^{(r)}),
	\end{eqnarray}	
	where the argument $\boldsymbol\Theta^{(r)}$ means that the conditional expectation is evaluated at the EM estimate at the $r$th iteration of the algorithm, and the underscript $\mathcal F_{t-1}$ meaning that the conditional expectation is taken with respect to the conditional density function $g(z_t|{\bf y}_t)=g({\bf y}_t,z_t)/p({\bf y}_t)$, where $g({\bf y}_t,z_t)\equiv p({\bf y}_t,z_t|\mathcal F_{t-1})$ and $p({\bf y}_t)\equiv p({\bf y}_t|\mathcal F_{t-1})$ is the joint probability function of ${\bf Y}_t$ given $\mathcal F_{t-1}$. The conditional expectations involved in the computation of (\ref{Qfunction}) are derived in the next proposition.

\begin{proposition}
		\label{prop:cond}
		Suppose that $\left\{{\bf Y}_t\right\}_{t\geq1}$ is a $\mbox{MPGIG}_p\mbox{-INGARCH}$ process with associated latent GIG random variables $\{Z_t\}_{t\geq1}$. Then,  $g(z_t|{\bf y}_t)$ is a density function from a three-parameter 		
		$\text{GIG}\left(2\displaystyle\sum_{i=1}^{p}\lambda_{it}+\phi,\phi,\displaystyle\sum_{i=1}^{p}y_{it}+\alpha\right)$	
		 distribution. Furthermore, the conditional expectations involved in (\ref{Qfunction}) are given by
		\begin{align*}
		\zeta_t&=\mathbb{E}_{\mathcal F_{t-1}}\left(Z_t\mid{\bf Y}_t\right)=		
		\left(\dfrac{\phi}{2\sum_{i=1}^{p}\lambda_{it}+\phi}\right)^{1/2}\dfrac{\mathcal{K}_{\sum_{i=1}^{p}y_{it}+\alpha+1}\left(\sqrt{\phi(2\sum_{i=1}^{p}\lambda_{it}+\phi)}\right)}{\mathcal{K}_{\sum_{i=1}^{p}y_{it}+\alpha}\left(\sqrt{\phi(2\sum_{i=1}^{p}\lambda_{it}+\phi)}\right)},\\
		\kappa_t&=\mathbb{E}_{\mathcal F_{t-1}}\left(Z_t^{-1}\mid{\bf Y}_t\right)=	\left(\dfrac{\phi}{2\sum_{i=1}^{p}\lambda_{it}+\phi}\right)^{-1/2}\dfrac{\mathcal{K}_{\sum_{i=1}^{p}y_{it}+\alpha-1}\left(\sqrt{\phi(2\sum_{i=1}^{p}\lambda_{it}+\phi)}\right)}{\mathcal{K}_{\sum_{i=1}^{p}y_{it}+\alpha}\left(\sqrt{\phi(2\sum_{i=1}^{p}\lambda_{it}+\phi)}\right)},\\
		\xi_t&=\mathbb{E}_{\mathcal F_{t-1}}\left(\log Z_t\mid {\bf Y}_t\right)=\dfrac{d}{d\alpha}\log\mathcal{K}_{\sum_{i=1}^{p}y_{it}+\alpha}\left(\sqrt{\phi\left(2\sum_{i=1}^{p}\lambda_{it}+\phi\right)}\right)+\dfrac{1}{2}\log\left(\dfrac{\phi}{2\sum_{i=1}^{p}\lambda_{it}+\phi}\right).		
		\end{align*}
\end{proposition}

\begin{proof}[Proof of Proposition \ref{prop:cond}] In what follows, the dependence of quantities on $\mathcal F_{t-1}$ is omitted for simplicity of notation. We have that the conditional density function $g(z_t|\boldsymbol{y}_t)$ is given by
	\begin{align*}
		g(z_t|\boldsymbol{y}_t)&=\dfrac{p(\boldsymbol{y}_t|z_t)\cdot g(z_t)}{p(\boldsymbol{y}_t)}\propto p(\boldsymbol{y}_t|z_t)\cdot g(z_t)\\
		&=\prod_{i=1}^p\dfrac{e^{-\lambda_{it}z_t}(\lambda_{it}z_t)^{y_{it}}}{y_{it}!}		
		\cdot \dfrac{1}{2\mathcal{K}_{\alpha}(\phi)}z_t^{\alpha-1}\exp\left(-\dfrac{\phi}{2}\left(z_t+ z_t^{-1}\right)\right)\\
		&\propto z_t^{\sum_{i=1}^py_{it}+\alpha-1}\exp\left\{-\dfrac{1}{2}\left[\left(2\sum_{i=1}^p\lambda_{it}+\phi\right)z_t+\phi z_t^{-1}\right]\right\},
	\end{align*}
	which is the kernel of a $\text{GIG}\left(2\displaystyle\sum_{i=1}^{p}\lambda_{it}+\phi,\phi,\displaystyle\sum_{i=1}^{p}y_{it}+\alpha\right)$ density function; see (\ref{densityGIG}). 
	
	Now, the conditional expectations stated in the proposition are obtained by using the fact that if $X\sim\mbox{GIG}(a,b,\alpha)$, then
	$E(X^u)=(b/a)^{u/2}\dfrac{\mathcal K_{\alpha+u}(\sqrt{ab})}{\mathcal K_{\alpha}(\sqrt{ab})}$, for $u\in\mathbb R$,	
	and
	$E(\log X)=\dfrac{d}{d\alpha}\log\mathcal K_{\alpha}(\sqrt{ab})+\dfrac{1}{2}\log(b/a)$;	
	for instance, see \cite{Jorgensen1982}.
\end{proof}
	
Using Proposition \ref{prop:cond}, the $Q$-function (\ref{Qfunction}) can be expressed by 
\begin{align} \label{eq:objfunction_split}
Q(\boldsymbol\Theta;\boldsymbol\Theta^{(r)})=Q_1(\phi,\alpha;\boldsymbol\Theta^{(r)})+Q_2(\boldsymbol\theta;\boldsymbol\Theta^{(r)}),
\end{align}  
 where
	\begin{align*}
		Q_1(\phi,\alpha;\boldsymbol{\Theta}^{(r)})&=\alpha\sum\limits_{t=2}^T  \xi_t^{(r)} - (T-1)\operatorname{log}\mathcal{K}_{\alpha}(\phi)-\cfrac{\phi}{2}\sum\limits_{t=2}^T(\zeta_t^{(r)}+\kappa_t^{(r)}),\\ 
		Q_2(\boldsymbol{\theta};\boldsymbol{\Theta}^{(r)}) &= \sum\limits_{t=2}^T \sum\limits_{i=1}^p y_{it}\operatorname{log}\left(\lambda_{it}\right)-\sum\limits_{t=2}^T \zeta_t^{(r)} \sum\limits_{i=1}^p\lambda_{it},
	\end{align*}
and $\zeta_t^{(r)}$, $\kappa_t^{(r)}$ and $\xi_t^{(r)}$ being the conditional expectations from Proposition  \ref{prop:cond} evaluated at $\boldsymbol\Theta=\boldsymbol\Theta^{(r)}$.

	It must be noted that the conditional expectation term $\xi_t$ is not in closed form due to the derivative of the Bessel function with respect to its order, which makes it hard to evaluate the $Q$-function. Even the computation of the other conditional expectations $\zeta_t$ and $\kappa_t$ might be hard to evaluate numerically in some cases due to a possible instability of the Bessel function. To solve these issues, we will employ a Monte Carlo EM (MCEM) method where the conditional expectations in the E-step are calculated via a Monte Carlo approximation. More specifically, for each $t=1,2,\cdots,T$, we first generate $m$ random draws $\left(Z_{t}^{(1)},Z_{t}^{(2)},\cdots,Z_{t}^{(m)} \right)$ from the conditional distribution of $Z_t$ given ${\bf Y}_t$ (see Proposition \ref{prop:cond}), and then we approximate the conditional expectations $\zeta_t$, $\kappa_t$ and $\xi_t$ respectively by $\widehat{\zeta}_t^{(r)}=\dfrac{1}{m}\sum\limits_{i=1}^{m}Z_{t}^{(i)}$, $\widehat{\kappa}_t^{(r)}=\dfrac{1}{m}\sum\limits_{i=1}^{m}\frac{1}{Z_{t}^{(i)}}$, and $\widehat{\xi}_t^{(r)}=\dfrac{1}{m}\sum\limits_{i=1}^{m}\log Z_{t}^{(i)}  $. More details on MCEM can be found in \cite{chan1995}.

	\subsubsection*{M-step}
	
	Since the objective function $Q(\boldsymbol{\Theta};\boldsymbol{\Theta}^{(r)})$ is decomposed into $Q_1(\phi,\alpha;\boldsymbol{\Theta}^{(r)})$ and $Q_2(\boldsymbol{\theta};\boldsymbol{\Theta}^{(r)})$, maximization of $Q(\boldsymbol{\Theta};\boldsymbol{\Theta}^{(r)})$ can be conducted separately as follows:
	\begin{align*}
	(\phi^{(r+1)},\alpha^{(r+1)})&=\operatorname*{argmax}_{\phi,\alpha}Q_1(\phi,\alpha;\boldsymbol{\Theta}^{(r)}),\\
	\boldsymbol{\theta}^{(r+1)}&=\operatorname*{argmax}_{\boldsymbol{\theta}}Q_2(\boldsymbol{\theta};\boldsymbol{\Theta}^{(r)}),
	\end{align*}
where the superscript $(r+1)$ denotes the EM-estimate of parameters at the $r+1$ iteration.
	Due to the lack of closed-form updating formula, a numerical optimization is needed for both maximizations. Maximizing $Q_1(\phi,\alpha)$ is implemented with typical numerical methods such as the \texttt{BFGS} algorithm using the \texttt{optim} function in \textsf{R}, but maximizing $Q_2(\boldsymbol{\theta})$ can seriously increase the computational burden since many parameters are to be updated. To alleviate this problem, 
	we borrow the idea of the generalized EM (GEM) algorithm \citep{Dempster1977}, 
	in which $\boldsymbol{\theta}^{(r+1)}$ is selected so that
		\begin{align}\label{eqn:GEM}
			Q(\boldsymbol{\theta}^{(r+1)}|\boldsymbol{\theta}^{(r)})\geq Q(\boldsymbol{\theta}^{(r)}|\boldsymbol{\theta}^{(r)}).
     	\end{align}
     
The GEM algorithm is attractive when maximization of $Q(\boldsymbol{\theta}^{(r)}|\boldsymbol{\theta}^{(r)})$ is computationally expensive. As a way to achieve (\ref{eqn:GEM}), \cite{lange1995} suggested the EM gradient algorithm where a nested loop in the M-step is simply replaced by one iteration of Newton's method. Because $Q_1 (\phi,\alpha;\boldsymbol{\Theta}^{(r)})$ is just of two dimensions and its derivative with respect to $\alpha$ is not in closed form, we only apply the one-step update to $Q_2 (\boldsymbol{\theta};\boldsymbol{\Theta}^{(r)})$. In other words, we have that
\begin{align*}
&	(\phi^{(r+1)},\alpha^{(r+1)})=\operatorname*{argmax}_{\phi,\alpha}Q_1(\phi,\alpha;\boldsymbol{\Theta}^{(r)}),\\
&	\boldsymbol{\theta}^{(r+1)}=\boldsymbol{\theta}^{(r)}-\textbf{H} (\boldsymbol{\theta}^{(r)})^{-1} \cdot \textbf{S} (\boldsymbol{\theta}^{(r)}),
\end{align*}  
where $\textbf{S}(\boldsymbol{\theta})=\cfrac{\partial Q_2 (\boldsymbol{\theta})}{\partial \boldsymbol{\theta}}=\sum\limits_{t=1}^n \cfrac{\partial\boldsymbol{\nu}_t^{\top}}{\partial\boldsymbol{\theta}}\big(\textbf{Y}_t-z_t \cdot \exp \big(\boldsymbol{\nu}_t \left(\boldsymbol{\theta}\right) \big) \big)$ is the gradient of $Q_2 (\boldsymbol{\theta})$ and $\textbf{H}(\boldsymbol{\theta})$ is its associate hessian matrix. In practice, we use the negative outer product of $\textbf{S}(\boldsymbol{\theta})$ as $\textbf{H}(\boldsymbol{\theta})$ instead of the hessian matrix.

The technical results connecting the EM gradient algorithm with the EM and the GEM algorithms, as discussed in \cite{lange1995}, are not guaranteed to hold in our proposed algorithm. However, we empirically confirmed through extensive simulation studies that the monotonic increase of the observed log-likelihood generally holds in practice, and the computational feasibility is also improved significantly. We refer to our proposed algorithm as the generalized Monte Carlo EM (GMCEM) from now on. We believe that the likehood increasing property deserves further research and could be addressed following ideas from \cite{cafetal2005}. The implementation of the MCEM and GMCEM algorithms are summarized in Algorithm \ref{alg:mcem}.

    \begin{algorithm}[ht!] 
	\SetAlgoLined
	\KwIn{Observed multivariate time series of counts $\{\textbf{Y}_t\}_{t=1}^T$}
	\begin{itemize}
		\item[1.] \textbf{(Initial guess)}: Input a starting value, denoted by $\boldsymbol{\Theta}^{(0)}$.
		\item[2.] \textbf{(MCE-step)}: For each $t=1,\cdots,T$, generate $m$ random draws of $Z_t$, say $(Z_{t}^{(1)},Z_{t}^{(2)},\cdots,Z_{t}^{(m)})$, with  $\boldsymbol{\Theta}=\boldsymbol{\Theta}^{(r)}$, where the superscript $(r)$ denotes the EM estimate of the parameter vector at the $r^{th}$ iteration of the loop, from the conditional distribution given in Proposition \ref{prop:cond}. \\
		Then, compute the approximate conditional expectations\\ $\widehat{\zeta}_t^{(r)}=\dfrac{1}{m}\displaystyle\sum\limits_{i=1}^{m}Z_{t}^{(i)}$, $\widehat{\kappa}_t^{(r)}=\dfrac{1}{m}\displaystyle\sum\limits_{i=1}^{m}\frac{1}{Z_{t}^{(i)}}$, and $\widehat{\xi}_t^{(r)}=\dfrac{1}{m}\displaystyle\sum\limits_{i=1}^{m}\log Z_{t}^{(i)}$. 
		\item[3.] \textbf{(M-step 1)}: Obtain $(\phi^{(r+1)},\alpha^{(r+1)})$ that maximizes $Q_1(\phi,\alpha;\boldsymbol{\Theta}^{(r)})$ with a numerical method.
		\item[4.] \textbf{(M-step 2 for MCEM)}: Obtain $\boldsymbol{\theta}^{(r+1)}$ that maximizes $Q_2(\boldsymbol{\theta};\boldsymbol{\Theta}^{(r)})$ with a numerical method.
		\item[] \textbf{(M-step 2 for GMCEM)}: Calculate $\boldsymbol{\theta}^{(r+1)}=\boldsymbol{\theta}^{(r)}-\textbf{H} (\boldsymbol{\theta}^{(r)})^{-1} \cdot \textbf{S} (\boldsymbol{\theta}^{(r)})$.
		\item[5.] Repeat Steps 2 to 4 until some pre-specified stopping criterion, for instance, $\max\limits_{i}|\boldsymbol{\Theta}_i^{(r+1)}-\boldsymbol{\Theta}_i^{(r)}|<0.001$, is achieved.
	\end{itemize}
	\caption{MCEM and GMCEM algorithms} \label{alg:mcem}
\end{algorithm}

It is well-known that EM algorithm does not automatically yield standard errors for estimates. 
Consequently, several methods have been proposed in the literature to obtain them such as Louis's method \citep{Louis1982} and the supplemented EM algorithm \citep{Meng1991}. Due to the fully parametric nature of our model, we employ a parametric boostrapping technique which is not only intuitive to implement, but is also useful in small and moderate sample size settings. Algorithm \ref{alg:pb} outlines the implementation of a parametric bootstrapping to obtain standard errors of EM estimates for our model.

\begin{algorithm}[ht!]
	\SetAlgoLined
	\KwIn{Parameter estimates, denoted by $\widehat{\boldsymbol{\Theta}}$, calculated from the given data.}
	\begin{itemize}
		\item[1.] \textbf{(Data generation)}: With the estimate vector $\widehat{\boldsymbol{\Theta}}$, generate a new realization of size $T$ using the MPGIG$_p$-INGARCH process from Definition \ref{def:MPGIGINGARCH}.
		\item[2.] \textbf{(Parameter estimation)}: With this generated realization, apply the MCEM/GMCEM from Algorithm \ref{alg:mcem} to obtain the bootstrap parameter estimates.
		\item[3.] Repeat Steps 1 and 2 $B$ times to obtain $\left(\widehat{\boldsymbol{\Theta}}^{(1)},\widehat{\boldsymbol{\Theta}}^{(2)},\cdots,\widehat{\boldsymbol{\Theta}}^{(B)}\right)$, where $B$ is the number of bootstrap replications.
		\item[4.] Compute the empirical standard errors using the $B$ bootstrap estimates from Step 3.
	\end{itemize}
	\caption{Parametric bootstrap to obtain standard errors}  \label{alg:pb}
\end{algorithm}

\subsection{Hybrid estimation method (H-GMCEM)}

When the dimension $p$ of the observed count time series increases, empirical evidence suggests that the MCEM/GMCEM algorithms described in Algorithm \ref{alg:mcem} have significant computational burden. To address this problem, we propose a hybrid version of the GMCEM algorithm referred to as H-GMCEM. The idea is to combine the quasi-maximum likelihood estimation technique and the GMCEM algorithm, and this is seen to reduce the computational burden. 

\noindent For the quasi-maximum likelihood estimation part of this approach, we only assume Poisson marginals without specifying the contemporaneous correlation structure, which is the same approach taken in \cite{fokianos2020}. More precisely, we have that
\begin{align}
\label{quasi}
Y_{it}|\mathcal{F}_{t-1} \sim \text{Poisson}(\lambda_{it})\text{, where }  \log \boldsymbol{\lambda_t} = \boldsymbol{\nu}_t=\textbf{d}^*+\textbf{A}\boldsymbol{\nu}_{t-1}+\textbf{B}\log(\textbf{Y}_{t-1}+\boldsymbol{1}_p),
 \end{align} 
 for $t=2,\cdots,T$, and $i=1,\cdots, p$. It should be noted that $\textbf{d}^*$ from (\ref{quasi}) is not the same as \textbf{d} from Definition \ref{def:MPGIGINGARCH}. Since $E(Y_{it}|\mathcal F_{t-1})=\lambda_{it}\mathcal R_\alpha(\phi)$ under our multivariate model, we obtain the following relationship $\textbf{d}^*=\textbf{d}+\log\mathcal R_\alpha(\phi)({\bf I}-\textbf{A})\boldsymbol{1}_p$.
  The quasi-likelihood and log-likelihood functions in terms of the parameter vector $\boldsymbol{\theta}^*\equiv (\textbf{d}^*,\mbox{vec}(\textbf{A}),\mbox{vec}(\textbf{B}))^\top$ are given by
\begin{align}
\label{quasi_lik}
&\mathcal{L}_{Q}(\boldsymbol{\theta}_0)=\prod\limits^{T}_{t=2}\prod\limits^{p}_{i=1}\left\{\cfrac{e^{-\lambda_{it}}\cdot\lambda_{it}^{y_{it}}}{y_{it}}\right\},\nonumber\\
&\ell_{Q}(\boldsymbol{\theta}_0)=-\sum\limits^{T}_{t=2}\sum\limits^{p}_{i=1}\lambda_{it}+\sum\limits_{t=2}^{T}\sum\limits_{i=1}^p y_{it}\log\lambda_{it}+\text{const.},
\end{align}
where the $\lambda_{it}$'s assume the form given in (\ref{quasi}). In the H-GMCEM algorithm, the parameter matrices $\textbf{A}$ and $\textbf{B}$ from Definition \ref{def:MPGIGINGARCH} are first estimated by maximizing the quasi log-likelihood in (\ref{quasi_lik}).  
We denote these quasi-maximum likelihood estimates of $(\textbf{A},\textbf{B})$ by $(\widetilde{\textbf{A}},\widetilde{\textbf{B}})$. Then, the other parameters are then estimated by the GMCEM method from Algorithm \ref{alg:mcem} by plugging $(\widetilde{\textbf{A}},\widetilde{\textbf{B}})$ into $(\textbf{A},\textbf{B})$; see Algorithm \ref{alg:h-gmcem}. Finite sample comparisons based on the performances and computing times of the GMCEM and the H-GMCEM estimators are discussed in Section \ref{sec:simulations}, especially with regards to increasing dimension $p$.\\

    \begin{algorithm}[ht!]
	\SetAlgoLined
	\KwIn{Observed multivariate time series of counts $\{\textbf{Y}_t\}_{t=1}^T$}
	\begin{itemize}
		\item[1.] \textbf{(Initial guess)}: Input the starting value parameter vector, denoted by $\boldsymbol{\Theta}^{(0)}$.
		\item[2.] \textbf{(QMLE-step)}: Obtain the quasi-maximum likelihood estimates $\widetilde{\textbf{A}}$ and $\widetilde{\textbf{B}}$ by\\ maximizing the quasi-log-likelihood given in (\ref{quasi_lik}).
		\item[3.] \textbf{(GMCEM-step)}: Implement the GMCEM method from Algorithm \ref{alg:mcem} (with $(\widetilde{\textbf{A}},\widetilde{\textbf{B}})$ plugged in) for estimating $\phi$, $\alpha$ and \textbf{d} with some pre-specified stopping criterion for\\
		the EM algorithm.
	\end{itemize}
	\caption{H-GMCEM algorithm} \label{alg:h-gmcem}
\end{algorithm}

	\section{Simulation studies}
	\label{sec:simulations}
	
	In this section, we conduct extensive Monte Carlo simulation studies to examine the finite-sample performance of the GMCEM and H-GMCEM algorithms proposed in Section \ref{sec:estimation}. Attempts were made to compare the results from the GMCEM method and the direct optimization of the conditional log-likelihood, but the latter was computationally expensive, especially for dimensions higher than 3. 
	
	For each simulation scheme described below, 500 replications were generated and the corresponding boxplots of the estimates are presented. For Schemes 1-3, sample sizes $T=200,500,1000$ are considered, and for Schemes 4-6, sample sizes $T=500,1000$ are taken up. The adopted stopping criterion for the EM-algorithms was $\max\limits_{i}|\boldsymbol{\Theta}_i^{(r+1)}-\boldsymbol{\Theta}_i^{(r)}|<0.001$, with the maximum number of iterations set to be 500. We generate trajectories from the MPGIG$_p$-INGARCH model under six scenarios as specified in what follows. 
	\begin{itemize}
		\item Scheme 1 ($p=2$) \\
		$\phi=0.5$, $\alpha=1.5$,
		$\mathbf{\rm d} = \begin{bmatrix}
		0 \\
		0 
		\end{bmatrix}$, 
		$\mathbf{A} = \begin{bmatrix}
		0.3 &0 \\
		0 & 0.25 
		\end{bmatrix}$, 
		$
		\mathbf{B}=\begin{bmatrix}
		0.4 & 0\\
		0 & 0.3
		\end{bmatrix}$.
		\item Scheme 2 ($p=2$)\\
		$\phi=0.5$, $\alpha=1.5$,
		$\mathbf{\rm d} = \begin{bmatrix}
		0 \\
		1 
		\end{bmatrix}$, 
		$\mathbf{A} = \begin{bmatrix}
		0.3 &0 \\
		0 & 0.25 
		\end{bmatrix}$, 
		$
		\mathbf{B}=\begin{bmatrix}
		0.4 & 0\\
		0 & 0.3
		\end{bmatrix}$.
		\item Scheme 3 ($p=2$)\\
		$\phi=0.5$, $\alpha=-1.5$,
		$\mathbf{\rm d} = \begin{bmatrix}
		1 \\
		1 
		\end{bmatrix}$, 
		$\mathbf{A} = \begin{bmatrix}
		0.3 &0 \\
		0 & 0.25 
		\end{bmatrix}$, 
		$
		\mathbf{B}=\begin{bmatrix}
		0.4 & 0\\
		0 & 0.3
		\end{bmatrix}$.
	\end{itemize}

	\begin{itemize}
		\item Scheme 4 ($p=4$)\\
		$\phi=0.5$, $\alpha=1.5$,
		$\mathbf{\rm d} = \begin{bmatrix}
		0.5 \\
		0.5\\
		1\\
		0.5
		\end{bmatrix}$, 
		$\mathbf{A} = \begin{bmatrix}
		0.35 &0 &0&0\\
		0 &-0.3 &0&0\\
		0 &0 &0.4&0\\
		0 &0 &0&-0.3
		\end{bmatrix}$, 
		$
		\mathbf{B}=\begin{bmatrix}
		-0.3 &0 &0&0\\
		0 &0.3 &0&0\\
		0 &0 &-0.3&0\\
		0 &0 &0&0.4
		\end{bmatrix}$.
		
		\item Scheme 5 ($p=4$)\\
		$\phi=0.5$, $\alpha=1.5$,
		$\mathbf{\rm d} = \begin{bmatrix}
		0.5 \\
		0.5\\
		1\\
		0.5
		\end{bmatrix}$, 
		$\mathbf{A} = \begin{bmatrix}
		0.35 &-0.2 &0&0\\
		0 &-0.3 &0&0\\
		0 &0 &0.4&0\\
		0 &0 &0.2&-0.3
		\end{bmatrix}$, 
		$
		\mathbf{B}=\begin{bmatrix}
		-0.3 &0.2 &0.0&0\\
		0 &0.3 &0&0\\
		0 &0 &-0.3&0\\
		0 &0&-0.25&0.4
		\end{bmatrix}$.
		\item Scheme 6 ($p=10$)\\
		$\phi=0.5$, $\alpha=1.5$,
		$\textbf{d}=\mbox{diag}(0,0,0,0,1,0,0.8,0.5,1,0.8)$, \\
		$\mathbf{A}=\mbox{diag}(0.30,0.20,0.35,0.40,-0.20,0.30,-0.15,0.15,-0.25,-0.10)$, \\
		$\mathbf{B}=\mbox{diag}(0.30,0.35,0.30,0.20,-0.20,0.40,-0.20,0.25,-0.15,-0.20)$.
	\end{itemize}
	
	Note that for Schemes 5 and 6, we assume that the matrices \textbf{A} and \textbf{B} are either diagonal or band-limited, which is a reasonable assumption when relevant information is given or when the dimension is large and the matrices are sparse. 
	
	In Figures \ref{fig:Bivar}-\ref{fig:Scheme6_constr_Q},  we present boxplots of the estimates obtained with the proposed GMCEM method. It can be seen that the bias and variance of the estimators decrease as the sample size increases in all simulation schemes, which empirically indicates the consistency of the proposed EM-algorithm based estimators. Note that for Scheme 6, in Figures \ref{fig:Scheme6_constr} and \ref{fig:Scheme6_constr_Q}, we present the estimation results for both the GMCEM and the H-GMCEM methods.

	Another issue is the running time for estimating the parameters since the number of parameters is proportional to $p^2$, as mentioned earlier.  To this end, we report a summary of computing times for each simulation scheme in Table \ref{tab:computing}. For Schemes 1-3, it takes less than a minute in most cases. When it comes to the larger dimensional case, i.e., Scheme 6, it is seen that placing a constraint on $\mathbf{A}$ and $\mathbf{B}$ greatly reduces the computational burden as expected. In Scheme 6, the improvements in the running time of the H-GMCEM method when compared to the GMCEM method is evident. 
	
	\begin{table}[ht!]
		\centering
		\begin{tabular}{c|ccc|ccc|ccc}
			\hline
			\multicolumn{1}{l}{$ $} & \multicolumn{3}{c}{Scheme 1}                                                       & \multicolumn{3}{c}{Scheme 2}                                                       & \multicolumn{3}{c}{Scheme 3}                                                       \\ \hline
			$T$	 & 200 & 500 & 1000 & 200 & 500 & 1000 & 200 & 500 & 1000 \\ \hline
			1st Qu. & 12.60 & 22.50 & 30.90 & 19.50 & 33.60 & 32.90 & 5.20 & 10.20 & 19.60 \\ 
			Median & 16.90 & 30.30 & 50.50 & 27.90 & 56.90 & 91.80 & 7.70 & 12.90 & 23.30 \\ 
			3rd Qu. & 22.10 & 41.30 & 70.50 & 40.50 & 79.30 & 146.70 & 17.90 & 18.30 & 28.90 \\ \hline
		\end{tabular}		
		\begin{threeparttable}
			\centering
			\begin{tabular}{c|cc|cc|cc|cc|cc|cc}
				\hline
				\multicolumn{1}{l}{} & \multicolumn{2}{c}{Scheme 4} & \multicolumn{2}{c}{Scheme 4$^c$}   & \multicolumn{2}{c}{Scheme 5} & \multicolumn{2}{c}{Scheme 5$^c$}& \multicolumn{2}{c}{Scheme 6$^c$} & \multicolumn{2}{c}{Scheme 6$^*$}    \\ \hline
				$T$& 500 & 1000 & 500 & 1000 & 500 & 1000 & 500 & 1000 & 500 & 1000 & 500 & 1000 \\ 
				\hline
				1st Qu. & 40.90 & 68.10 & 8.40 & 16.20 & 69.60 & 114.00 & 13.70 & 24.10 & 13.40 & 25.10 & 1.60 & 3.50 \\ 
				Median & 53.40 & 84.30 & 10.40 & 20.10 & 87.20 & 146.70 & 24.50 & 42.40 & 17.50 & 31.60 & 2.30 & 3.60 \\ 
				3rd Qu. & 72.40 & 111.50 & 12.60 & 24.10 & 110.40 & 175.90 & 36.90 & 67.60 & 20.60 & 40.10 & 2.40 & 5.30 \\ 
				\hline
			\end{tabular}
			\begin{tablenotes}\footnotesize
				\item [c]GMCEM method with band-limiting constraints on $\textbf{A}$ and $\textbf{B}$.
				\item [*]H-GMCEM method: hybrid estimation using EM algorithm and quasi-likelihood method.	
			\end{tablenotes}
		\end{threeparttable}
		\caption{Descriptive analysis of the computing times (in seconds) of the GMCEM and H-GMECM methods for each simulation scheme.}	\label{tab:computing}
	\end{table}

	\begin{sidewaysfigure}[htbp]
		\centering
		\includegraphics[width=9.5in,height=7in]{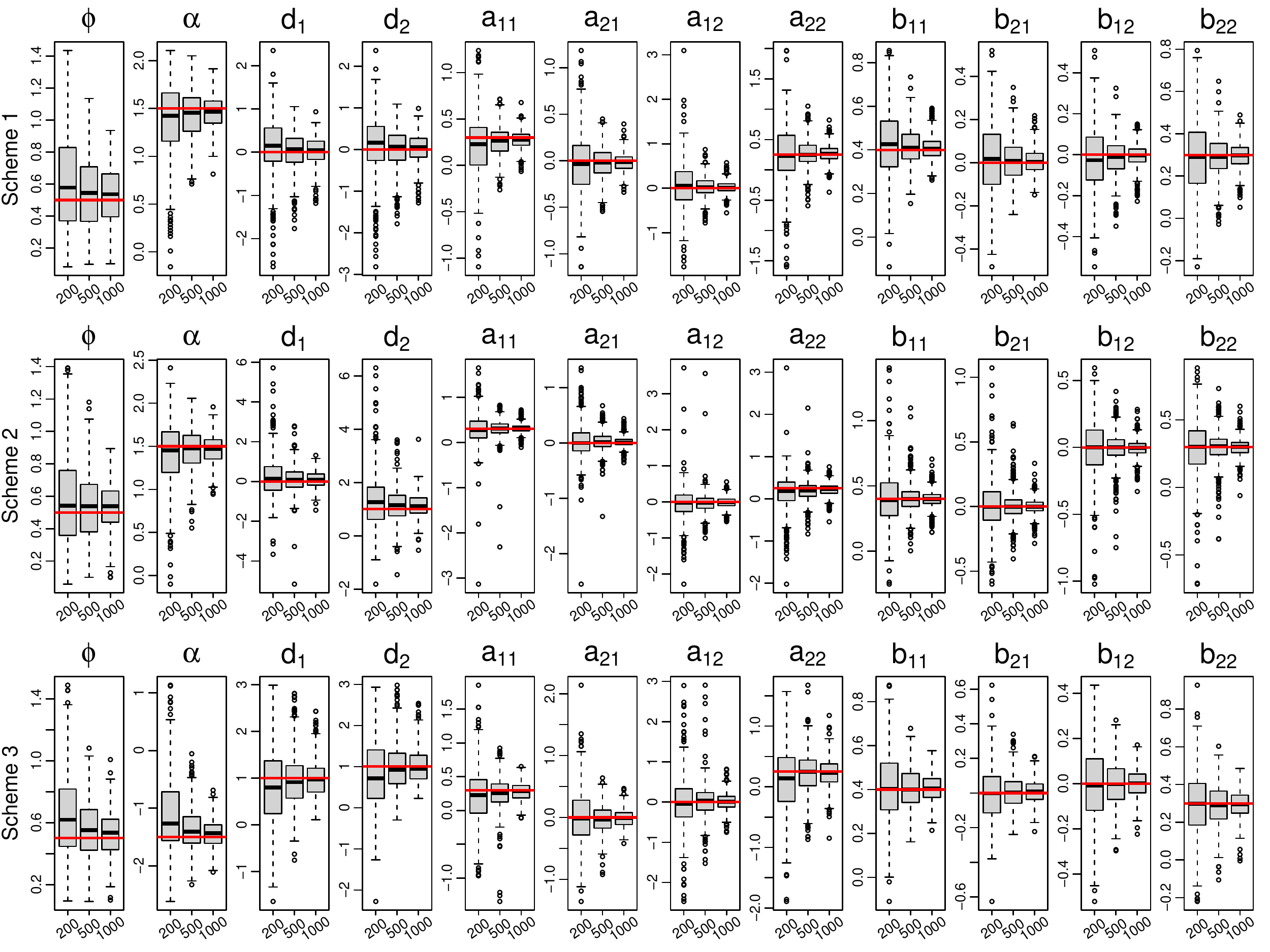}
		\caption{Boxplots of GMCEM and H-GMCEM estimates. Solid red lines represent the true parameter values.}
		\label{fig:Bivar}
	\end{sidewaysfigure}
	
	\begin{sidewaysfigure}[htbp]
		\centering
		\includegraphics[width=9.5in,height=7in]{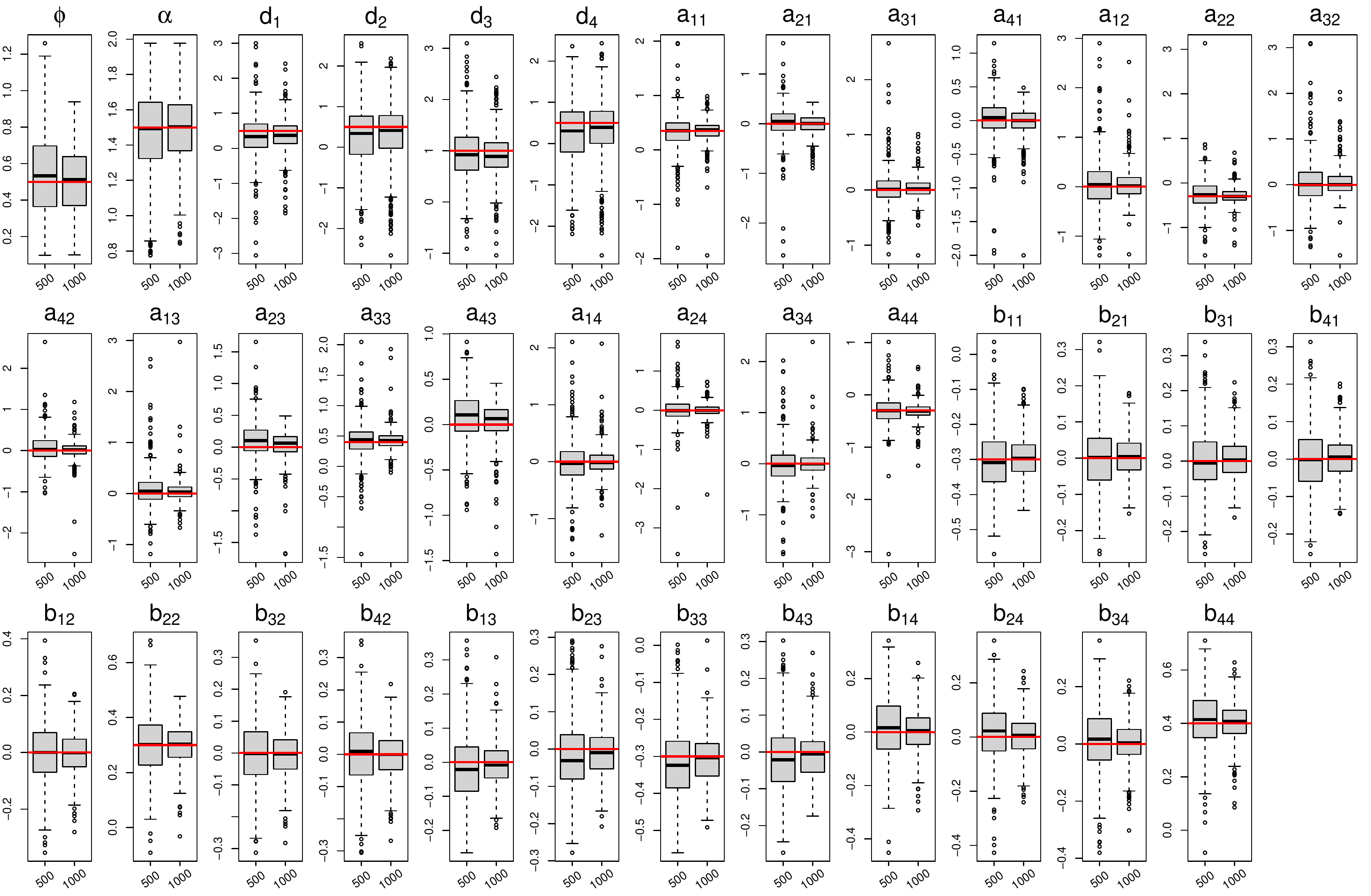}
		\caption{Boxplots of GMCEM estimates for Scheme 4. Solid red lines represent the true parameter values.}
		\label{fig:Scheme4}
	\end{sidewaysfigure}
	
	\begin{sidewaysfigure}[htbp]
		\centering
		\includegraphics[width=9.5in,height=7in]{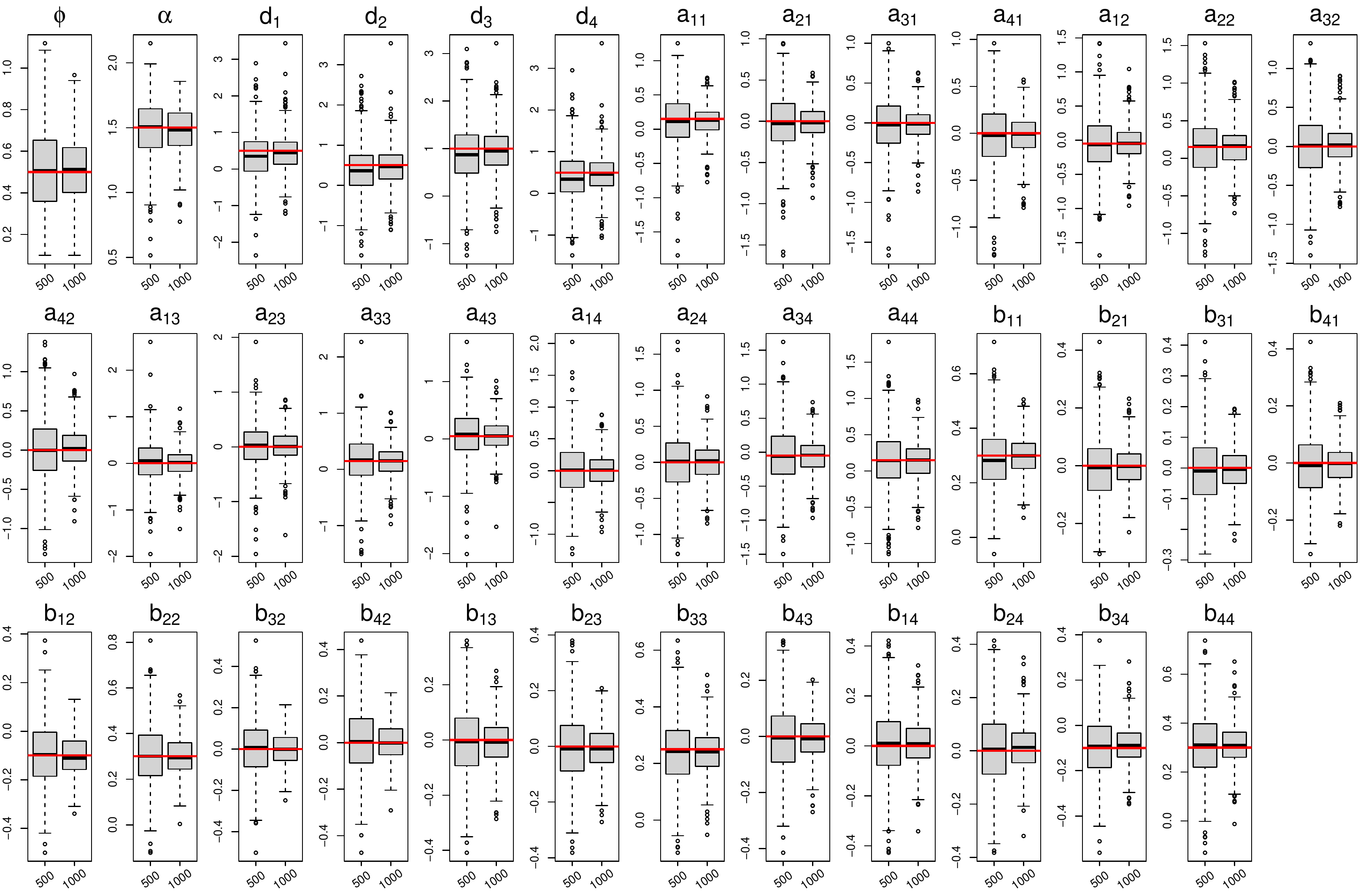}
		\caption{Boxplots of GMCEM estimates for Scheme 5. Solid red lines represent the true parameter values. No diagonal/band constraints on  $\textbf{A}$ and $\textbf{B}$.}
		\label{fig:Scheme5}
	\end{sidewaysfigure}

	\begin{sidewaysfigure}[htbp]

		\centering
		\includegraphics[width=9.5in,height=2.2in]{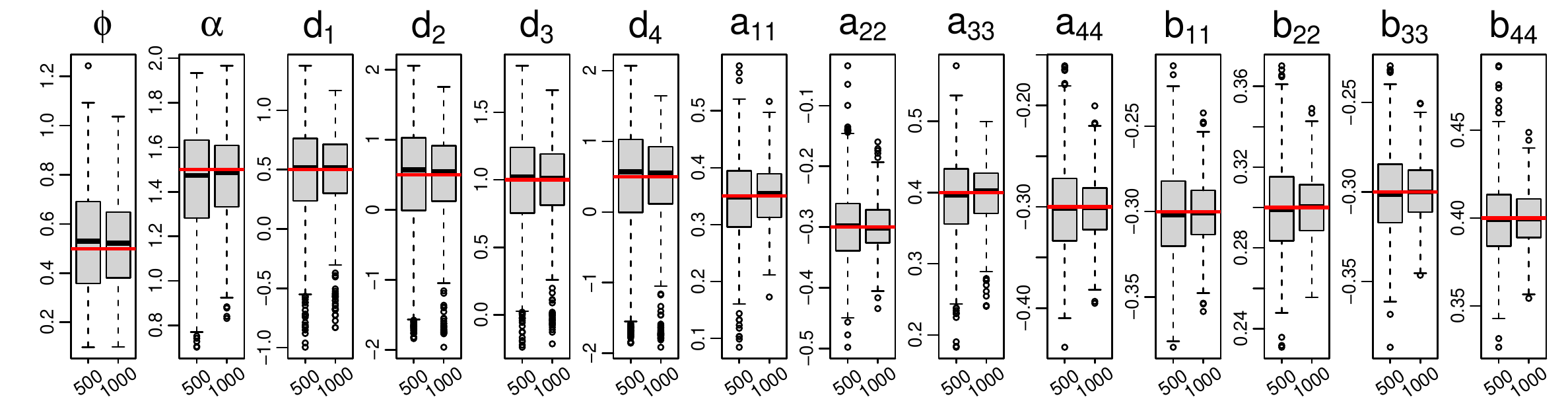}
		\caption{Boxplots of ML estimates for Scheme 4.  Matrices $\textbf{A}$ and $\textbf{B}$ are assumed to be diagnoal.}
			\label{fig:Scheme45_const}
		\centering
		\includegraphics[width=9.5in,height=4.4in]{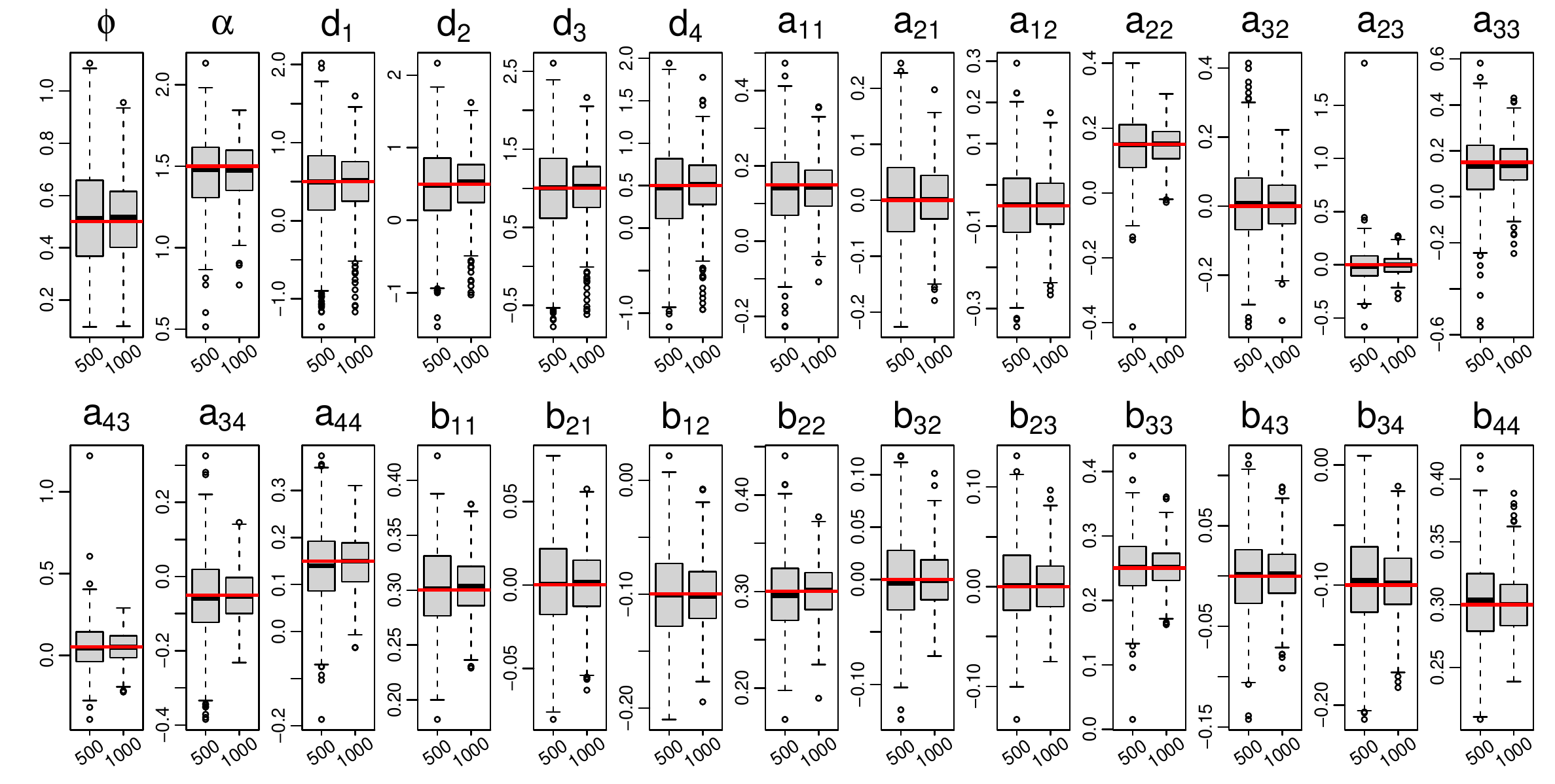}
		\caption{Boxplots of GMCEM estimates for Scheme 5. Matrices  $\textbf{A}$ and $\textbf{B}$ are assumed to be a band-limited.}
\end{sidewaysfigure}

		

	\begin{sidewaysfigure}[htbp]
	\centering
	\includegraphics[width=9.5in,height=7in]{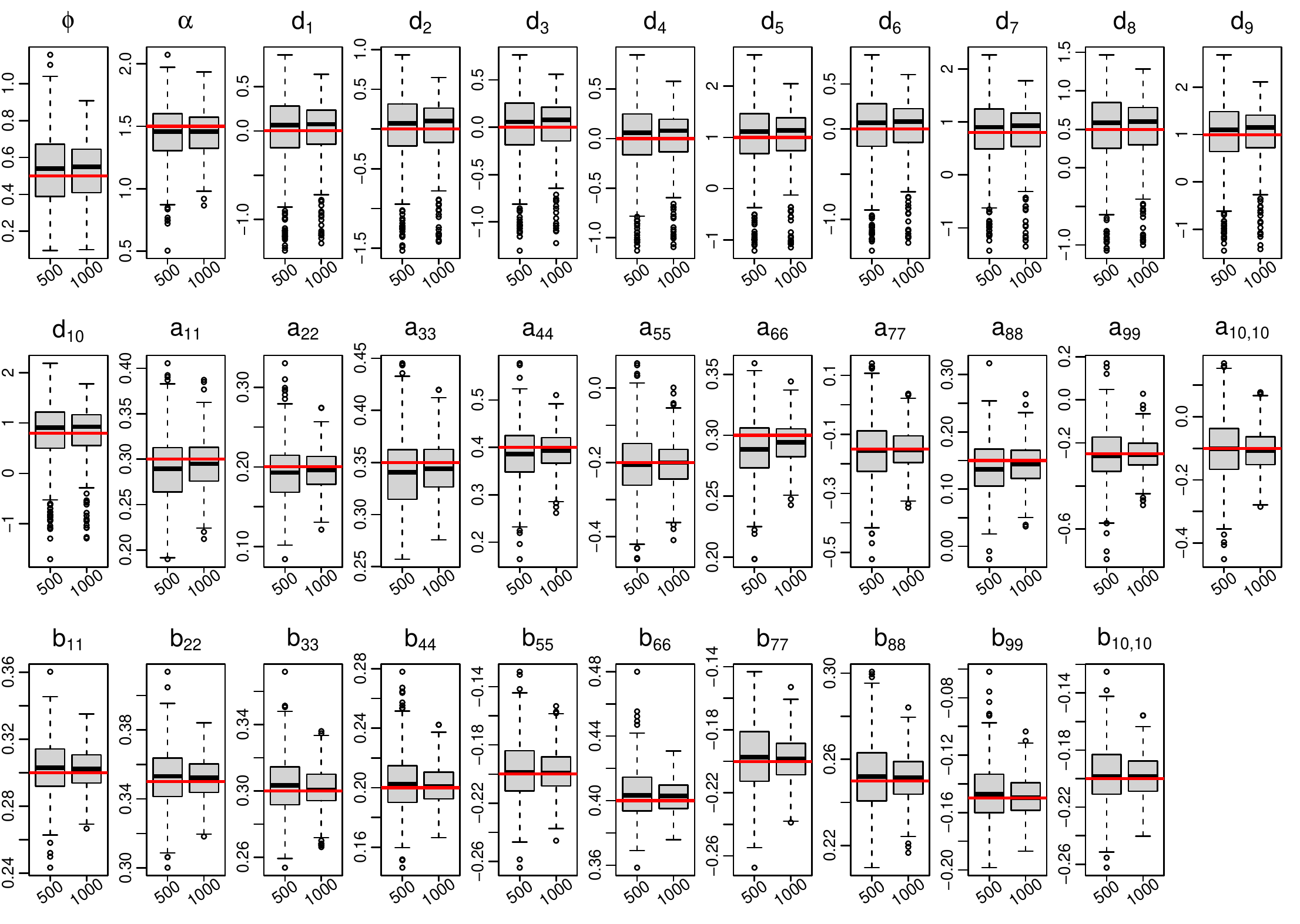}
	\caption{Boxplots of GMCEM estimates for Scheme 6.  Matrices  $\textbf{A}$ and $\textbf{B}$ are assumed to be diagonal.}
	\label{fig:Scheme6_constr}
\end{sidewaysfigure}
	
	\begin{sidewaysfigure}[htbp]
	\centering
	\includegraphics[width=9.5in,height=7in]{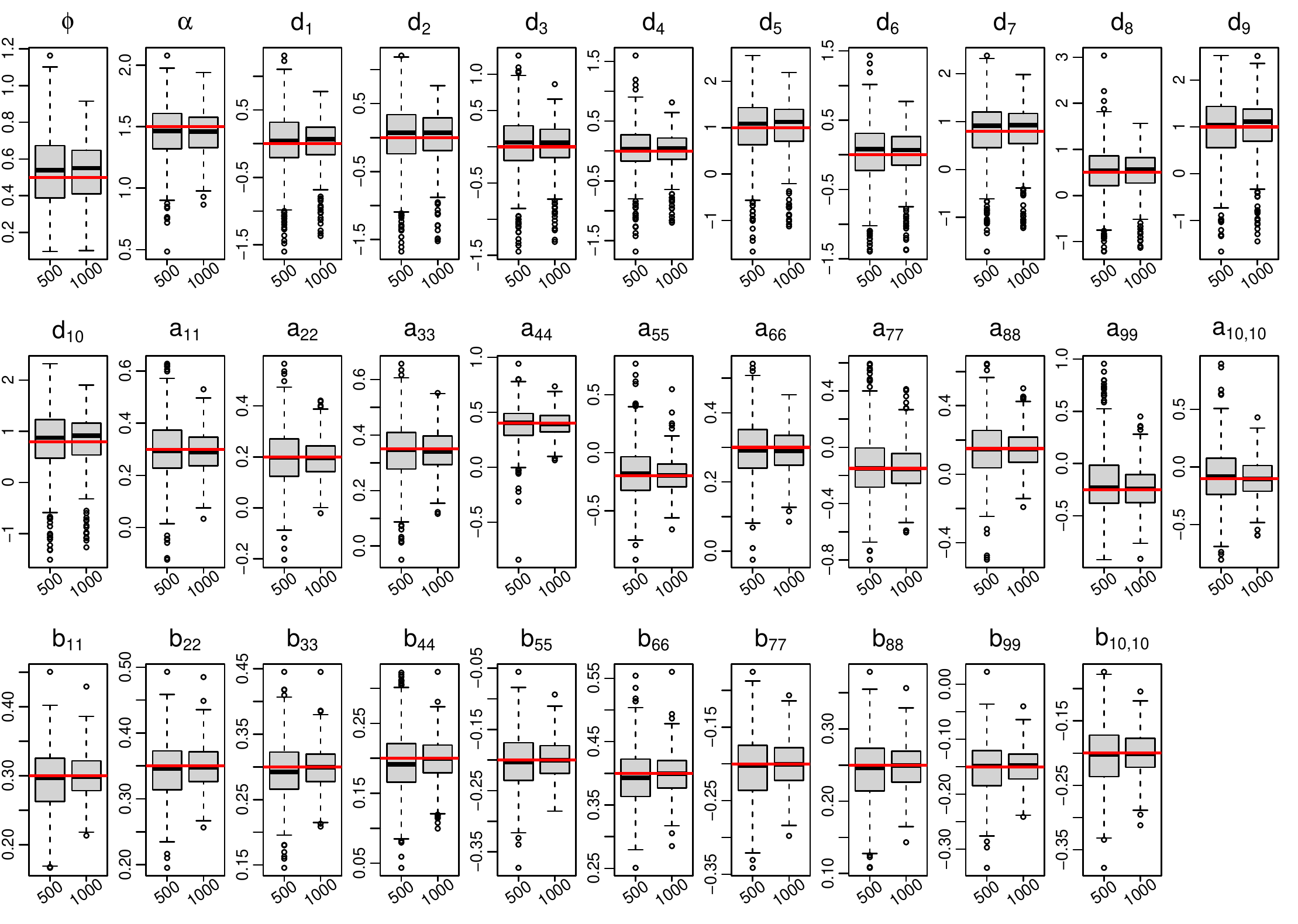}
	\caption{Boxplots of H-GMCEM method estimates for Scheme 6.  Matrices  $\textbf{A}$ and $\textbf{B}$ are assumed to be diagonal.}
	\label{fig:Scheme6_constr_Q}
\end{sidewaysfigure}

 	\section{Application to modeling monthly crime data}
	\label{sec:application}

	In this section, we present an application of the proposed method in modeling monthly number of drug-related crimes in the state of New South Wales (NSW), Australia. Model adequacy checks are done using probability integral transformation (PIT) plots, and comparisons with the existing multivariate Poisson-INGARCH model by \cite{fokianos2020} are addressed.

We consider monthly counts of cannabis possession-related offenses in the Greater Newcastle (GNC) and Mid North Coast (MNC) regions of NSW from January 1995 to December 2011. This yields a bivariate count time series with a total of 204 observations, which are presented in Figure \ref{fig:tsplot}. Also, Figure \ref{fig:map} provides a geographical plot indicating the two Australian regions MNC and GNC. For each of the two count time series, we inspect the presence of heavy-tailed behavior using the tail index (TI) \citep{Qian2020} defined as $\mbox{TI}=\widehat{\gamma}-\widehat{\gamma}_{nb}$, where $\widehat{\gamma}_{nb} = \dfrac{2\widehat{\sigma}^2-\widehat{\mu}}{\widehat{\mu}\widehat{\sigma}}$ is the index with respect to the negative binomial distribution, and $\widehat{\gamma}$, $\widehat{\sigma}^2$, and $\widehat{\mu}$ represent the sample skewness, variance, and mean, respectively. A positive value of $\mbox{TI}$ indicates a heavier tail compared to the negative binomial distribution. The estimated tail index values for the two count time series are 0.334 (MNC) and 0.271 (GNC). Therefore, we have evidence of heavy-taildness on this dataset.  The sample autocorrelation function (ACF) plots, provided in Figure \ref{fig:acf}, indicate that the observed series of counts in both regions involve significant temporal correlations along with a hint of seasonality with a period of roughly 12 months. Further, the cross-correlation function (CCF) plot at the bottom of Figure \ref{fig:acf} shows presence of serial and contemporaneous  cross-correlation between the two series. Thus, a multivariate analysis is required rather than two separate univariate analyses.

	\begin{figure}[ht!]
		\centering
		\includegraphics[width=15cm]{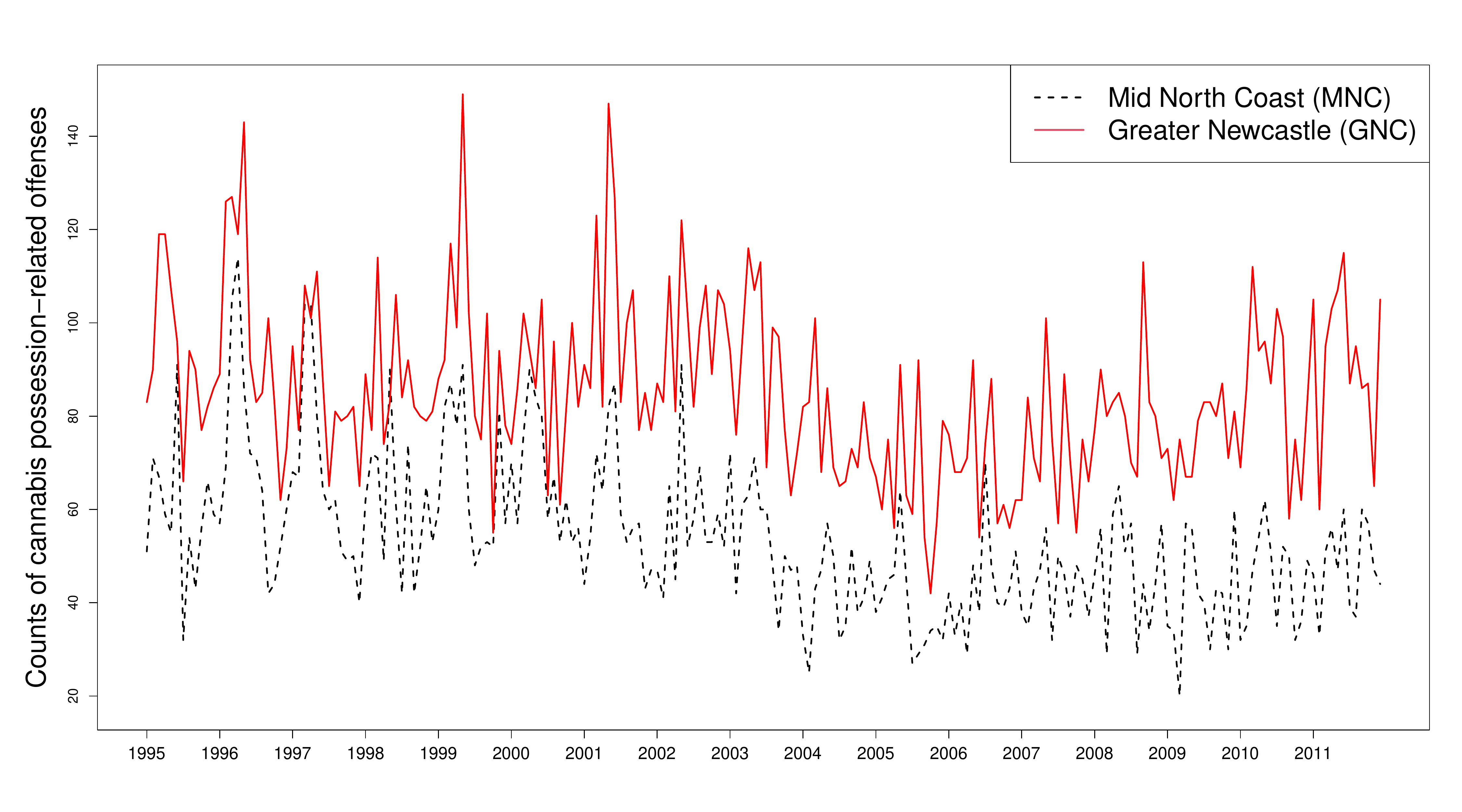}
		\caption{Time series plot of monthly counts of cannabis possession-related offenses in the two regions of MNC and GNC during the period January 1995 to December 2011.}
		\label{fig:tsplot}
				\includegraphics[width=10cm]{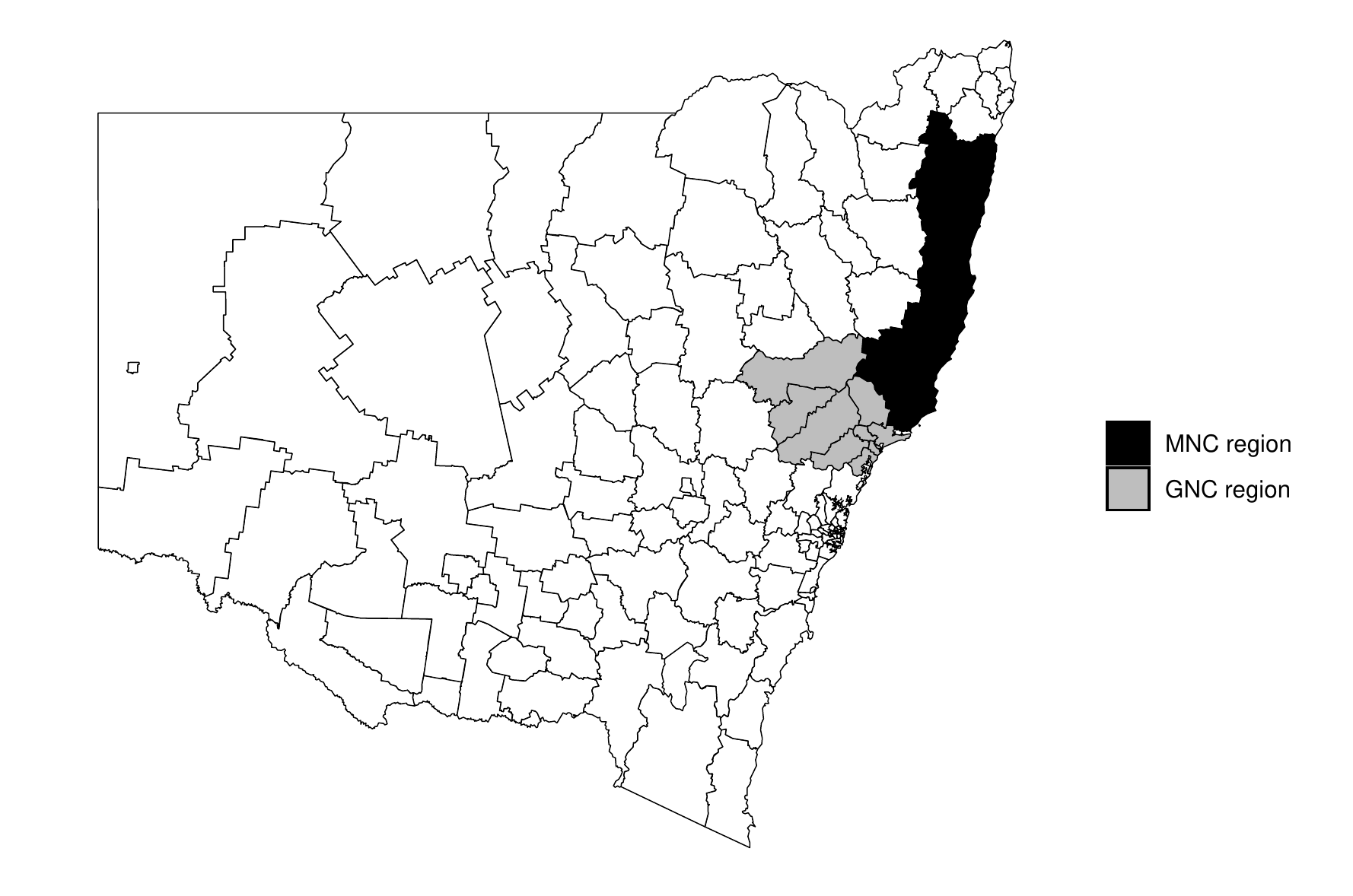}
		\caption{Map of the state of New South Wales (Australia) with the MNC and GNC regions highlighted in black and gray, respectively.}
		\label{fig:map}
	\end{figure}

	\begin{figure}[ht!]
		\centering
	\includegraphics[width=15cm]{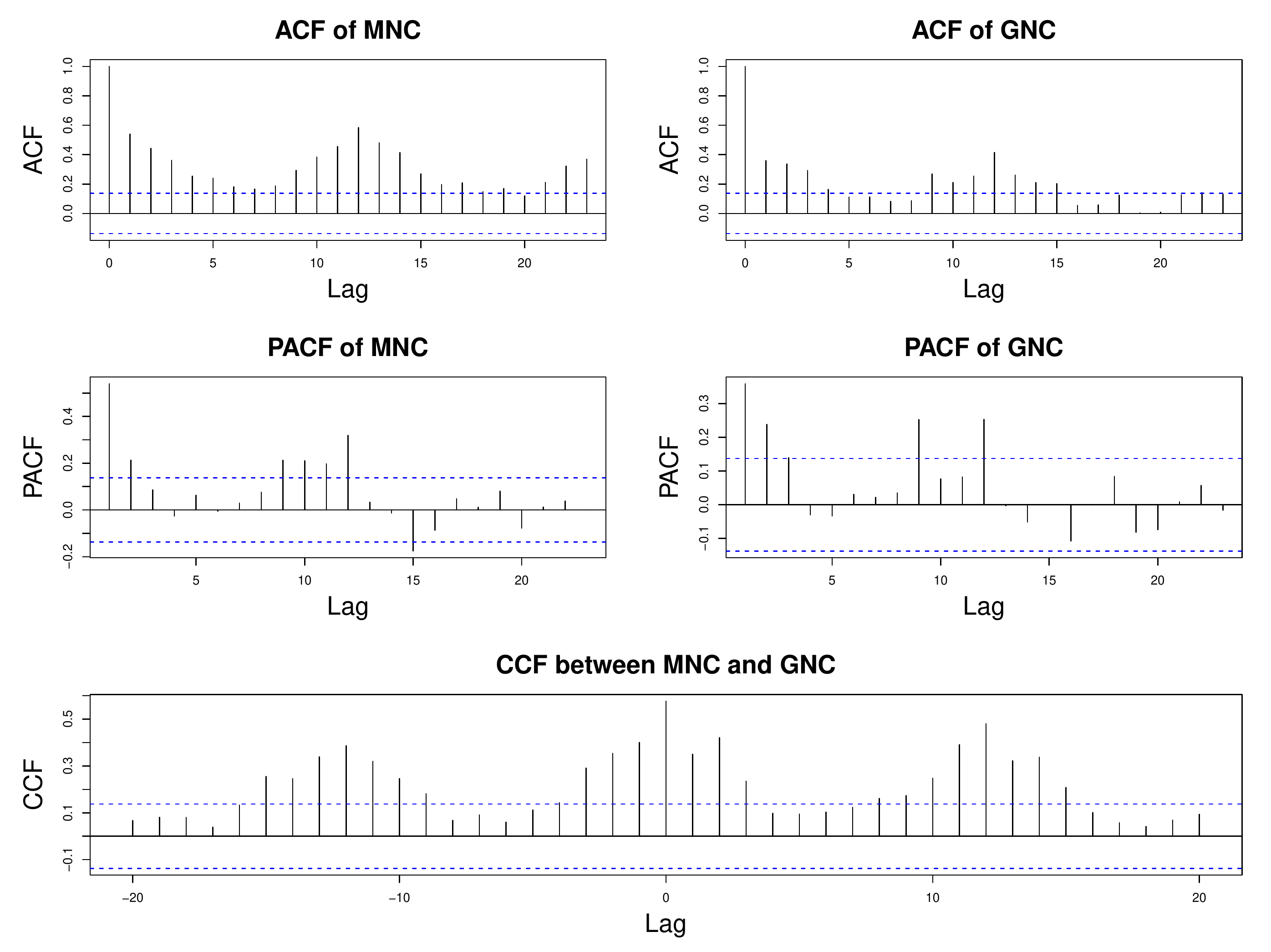}
	\caption{ACF (top), PACF (middle) and CCF (bottom) plots for monthly counts of cannabis possession-related offenses in the two regions of MNC and GNC.} 	\label{fig:acf}
   \end{figure}

	To adjust for the annual seasonal effect, the MPGIG$_p$-INGARCH model in Definition \ref{def:MPGIGINGARCH} is extended to include additional time lags into the mean equation from  (\ref{eqn::meanprocess}), which leads to the following formulation:
	\begin{align} \label{eq:extended_mean_equation}
		\boldsymbol{\nu}_{t} = \textbf{d}+\sum_{j\in I_1}\textbf{A}_{(j)}\boldsymbol{\nu}_{t-j} +\sum_{k\in I_2}\textbf{B}_{(k)}\log(\textbf{Y}_{t-k} +\boldsymbol{1}_p),\text{ for }t\geq 1,
	\end{align}
where $I_1$ and $I_2$ are index sets of time lags for past means and past observations, respectively. The ACF plots in Figure \ref{fig:acf} indicate seasonality with a period of approximately 12 months, and this prompts the inclusion of a time lag of 12 months in the mean equation in \eqref{eq:extended_mean_equation}. As outlined in Table \ref{tab:modelselection}, the MPGIG$_p$-INGARCH model that only considers time lags of 1 and 12 months with respect to the past observations (i.e. Model 4) is selected as the final model according to both Akaike information criteria (AIC) and Bayesian information criteria (BIC). 
	
	Table \ref{tab:estimate} presents the parameter estimates obtained using the MCEM algorithm described in Algorithm \ref{alg:mcem}, and the standard errors computed using a parametric bootstrap approach, described in Algorithm \ref{alg:pb}, with 500 bootstrap replications. The histogram of the bootstrap estimates of the parameters are provided in Figure \ref{fig:pb}. It can be observed that they appear to track close to the shape of a normal distribution.

	\begin{table}[ht!]
		\centering
		\begin{tabular}{c|c|c|c|c|c}
			\hline
			\multirow{2}{*}{Model}  &\multicolumn{2}{c|}{Time lag index sets}  &\multirow{2}{*}{BIC} &\multirow{2}{*}{AIC}& \multirow{2}{*}{\# of parameters}\\ \cline{2-3}
			&$I_1$&$I_2$& && \\
			\hline
			1 &1&1          &1731.242 & 1691.484& 12\\
			2 &1&1, 12      &1622.824 & 1570.704 & 16\\
			3 &1, 12& 1, 12 &1640.821 & 1575.671 & 20\\
			4 &$\cdot$ &1, 12 &\textbf{1604.377} & \textbf{1565.287}& \textbf{12}\\
			\hline
		\end{tabular}
		\caption{Comparison between candidate models using AIC and BIC.}
		\label{tab:modelselection}
	\end{table}

\begin{table}[ht!]
	\centering
	\begin{tabular}{c|c|cc|cc}
		\hline
		$\phi$ &$d_1$ & $b_{11(1)}$  &$b_{12(1)}$   &$b_{11(12)}$&$b_{12(12)}$\\ \hline
 \textbf{49.110} & \textbf{0.365} & \textbf{0.263} & \textbf{0.329} & \textbf{0.084} & \textbf{0.202} \\
 {\small (7.255)} & {\small (0.422)} & {\small (0.057)} & {\small (0.058)} & {\small (0.061)} & {\small (0.065)} \\
{\small $(38.996 , 66.033)$} & {\small $(-0.068 , 1.52)$} & {\small $(0.125 , 0.34)$} & {\small $(0.196 , 0.422)$} & {\small $(-0.035 , 0.202)$} & {\small $(0.048 , 0.302)$} \\ \hline
$\alpha$ & $d_2$&$b_{21(1)}$&$b_{22(1)}$ &$b_{21(12)}$&$b_{22(12)}$ \\ \hline
$\boldsymbol-$\textbf{1.158} & \textbf{1.911} & \textbf{0.046} & \textbf{0.069} & \textbf{0.218} & \textbf{0.253} \\
{\small (0.225)} & {\small (0.398)} & {\small (0.034)} & {\small (0.036)} & {\small (0.063)} & {\small (0.062)} \\
{\small $(-1.237 , -0.357)$} & {\small $(1.609 , 3.066)$} & {\small $(-0.032 , 0.102)$} & {\small $(0.002 , 0.138)$} & {\small $(0.069 , 0.319)$} & {\small $(0.078 , 0.314)$} \\ \hline
	\end{tabular}
	\caption{Parameter estimates for the MPGIG$_2$-INGARCH model with time lags of 1 and 12 months (Model 4) based on the MCEM algorithm. Standard errors and $95\%$ confidence intervals obtained via a parametric bootstrap are provided in the second and third lines, respectively.} 
	\label{tab:estimate}
\end{table}

    \begin{figure}[ht!]
    	\centering
		\includegraphics[width=16cm]{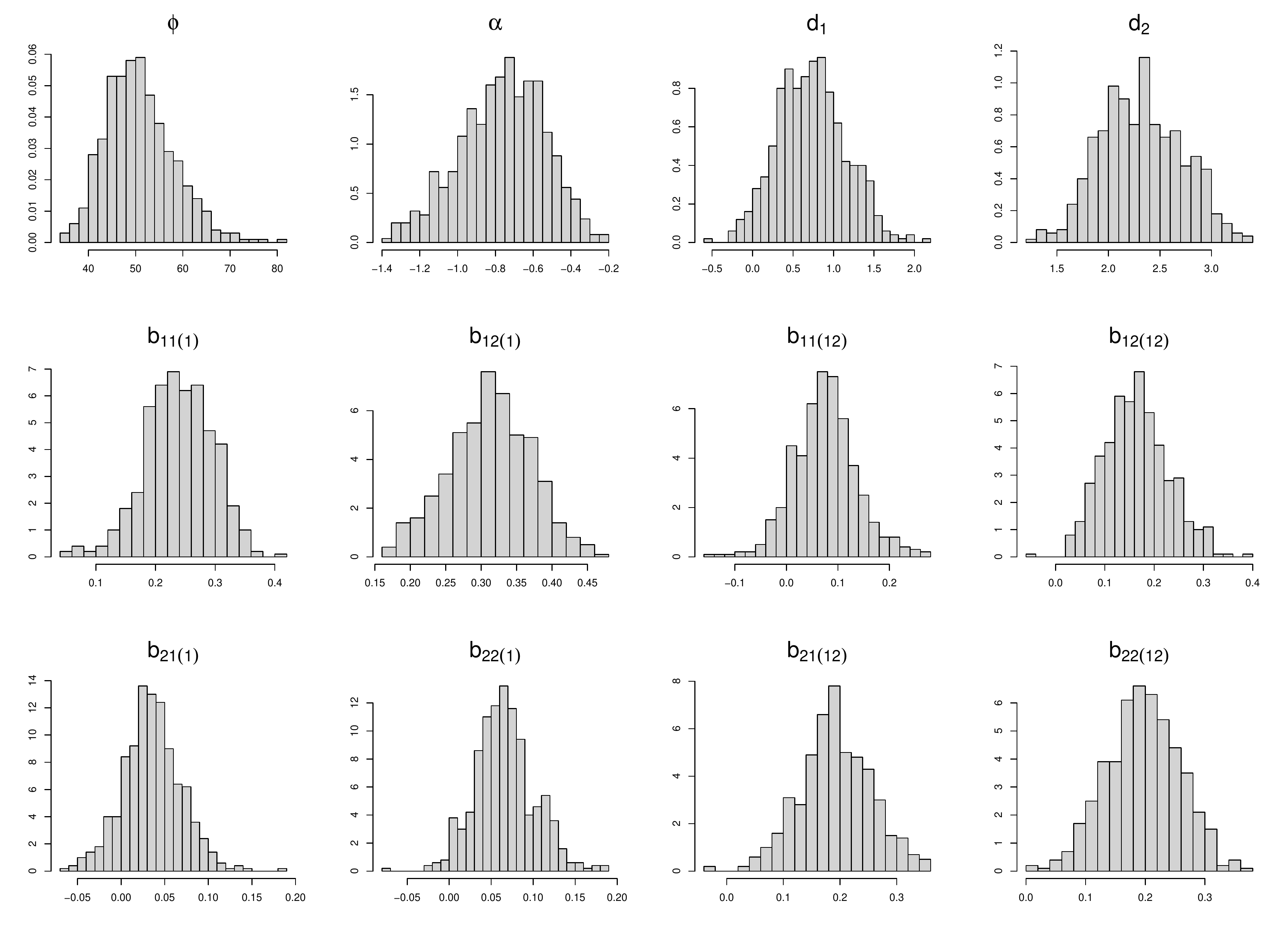}
		\caption{Bootstrap distribution of the parameter estimators from the MPGIG$_2$-INGARCH model model. Estimators are obtained using the MCEM algorithm.}
		\label{fig:pb}
	\end{figure}
	
	ACF and PACF plots of Pearson residuals from the fitted model is presented in Figure \ref{fig:resid} to assess the goodness-of-fit. At each time lag greater than 0, both the auto and cross-correlations fall within the confidence intervals, indicating that the fitted mean process adequately accounts for the temporal dependence structure. For checking the contemporaneous correlation structure, a $95\%$ confidence interval for the contemporaneous correlation parameter $\phi$ was computed using the parametric bootstrap technique. The estimated $\widehat{\rho}$ (=0.577) falls well within the confidence interval of (0.533,0.738), which further supports the selected model.
	
As our final model adequacy check, the non-randomized probability integral transformation (PIT) plot  \citep{Czado2009} is obtained for each count time series based on our approach and that one by \cite{fokianos2020}; see Figure \ref{fig:pit}. Compared to the multivariate Poisson-INGARCH model proposed by \cite{fokianos2020}, the PIT plots for our model are much closer to the uniform distribution, implying a much better and adequate goodness-of-fit. This establishes the advantage and ability of our proposed MPGIG$_p$-INGARCH model in being able to handle heavy-tailed behavior in multivariate count time series data. 
	
	\begin{figure}[ht!]
		\centering
		\includegraphics[width=15cm]{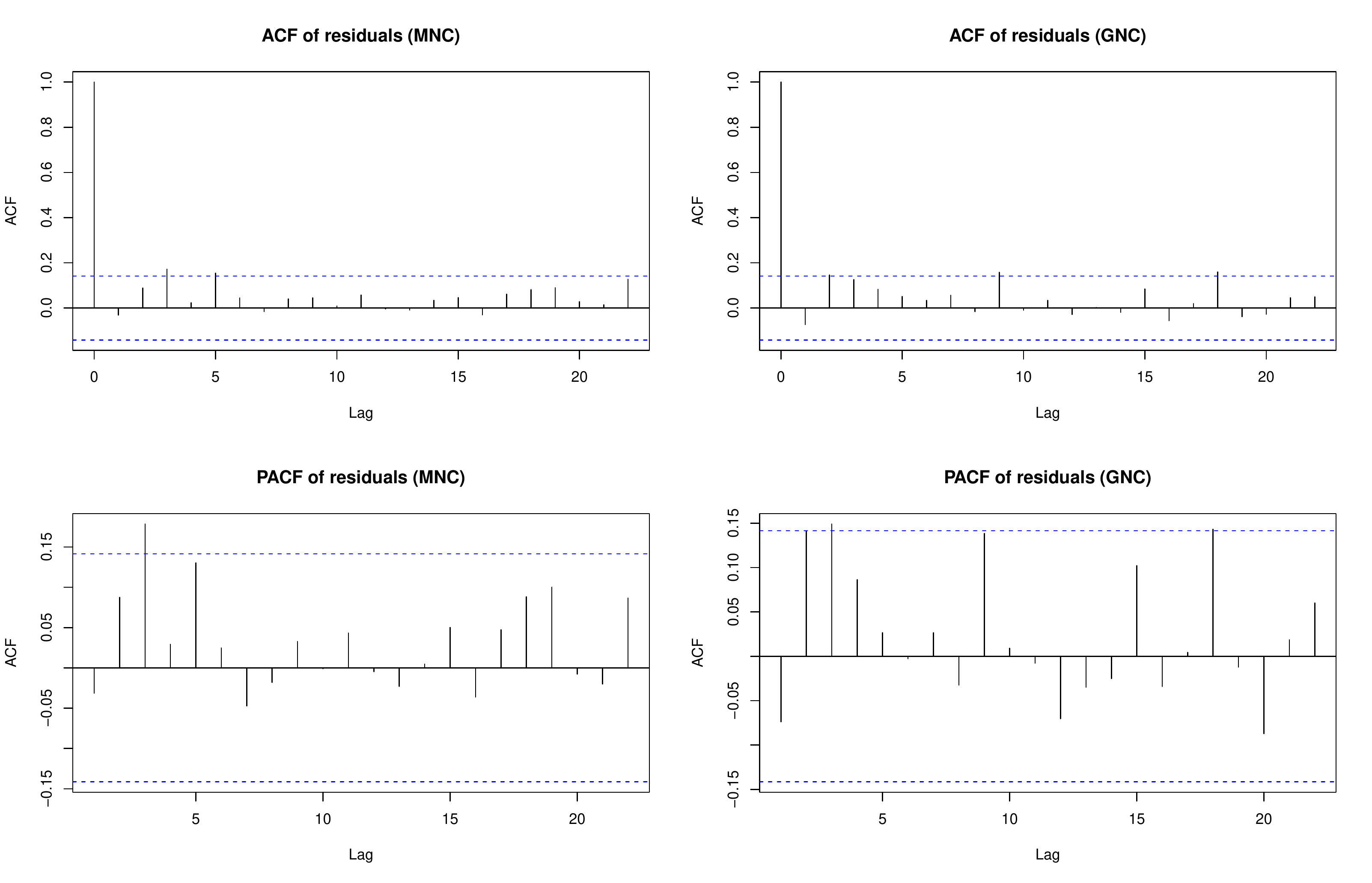}
		\caption{ACF and PACF of residuals from the fitted model.}
		\label{fig:resid}
	\end{figure}

	\begin{figure}[ht!]
	\centering
	\includegraphics[width=13cm]{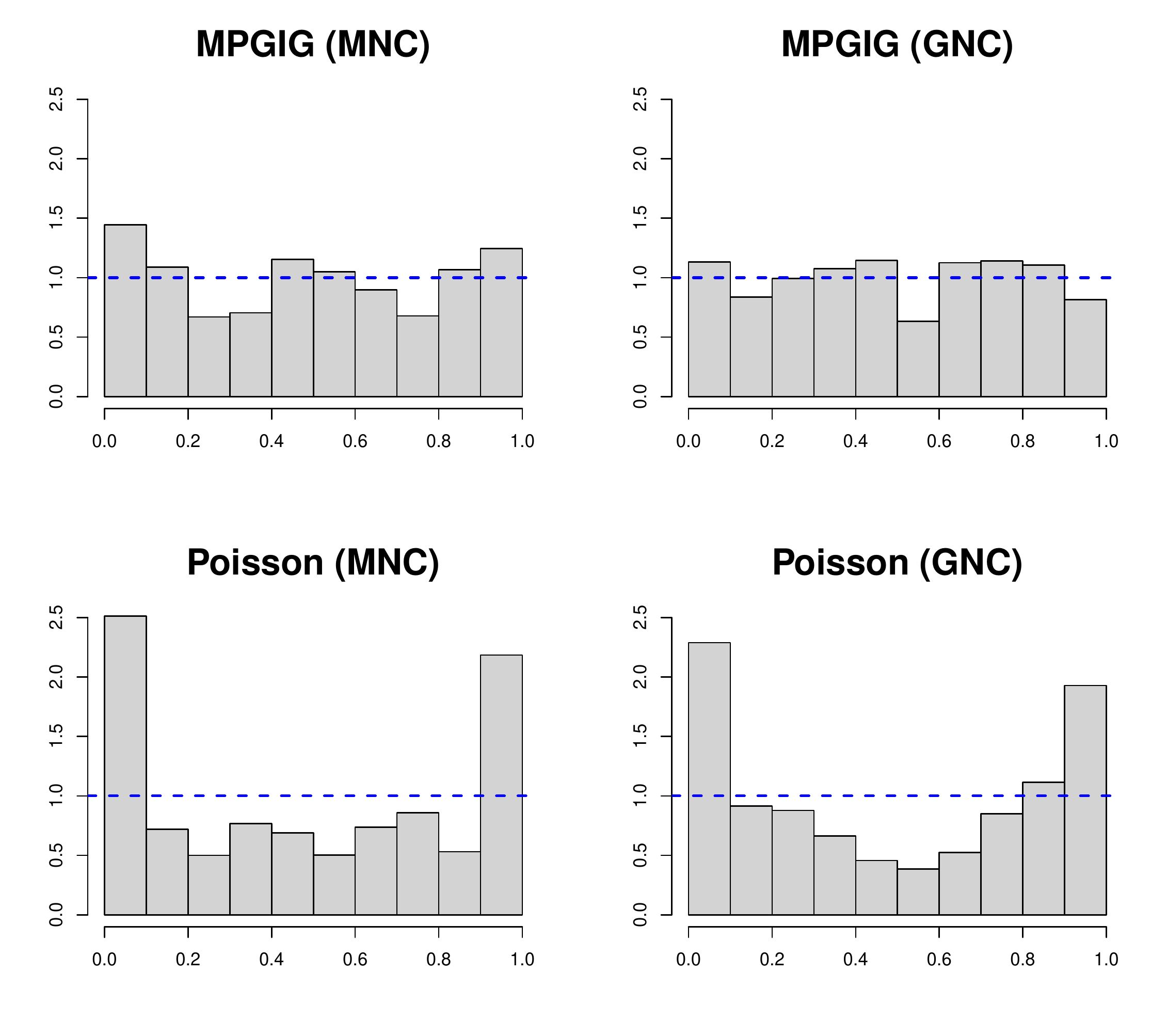}
	\caption{Probability integral transformation (PIT) plots of the proposed MPGIG$_p$-INGARCH model (top), and the multivariate Poisson-INGARCH model from \cite{fokianos2020} (bottom) for each count time series.}
	\label{fig:pit}
\end{figure}

	\section{Conclusion}
	\label{sec:conclusion}
In this article, we proposed a new multivariate count time series model that is capable of handling heavy-tailed behavior in count processes.  The new model, referred to MPGIG$_p$-INGARCH, was defined using a multivariate version of the INteger-valued Generalized AutoRegressive Conditional Heteroscedastic (INGARCH) model. The setup involved a latent generalized  inverse-Gaussian (GIG) random variable that controls contemporaneous dependence in the observed count process.	

A computationally feasible variant of the expectation-maximization (EM) algorithm, referred to as GMCEM, was developed. Aiming at cases when the dimension of the observed count process is large, a variant of the GMCEM algorithm that combines the quasi-maximum likelihood estimation method and the GMCEM method was also outlined. This variant, called H-GMCEM algorithm, is seen to have improved computing time performance in finite sample cases. In the many finite-sample simulation schemes that are considered, the proposed GMCEM and H-GMCEM methods shows increasing accuracy with increasing sample size. 

An application of the proposed method in modeling bivariate count time series data on cannabis possession-related offenses in Australia was explored. Empirical evidence indicated that this bivariate count time series exhibits heavy-tailed behavior and we illustrated, through suitable  model adequacy checks, that the proposed MPGIG$_p$-INGARCH model effectively handles such types of multivariate count time series data. 

Some points deserving future research are the inclusion of covariates, model extension allowing for general orders, and \texttt{R} implementation of our method with a general class of mixed Poisson distributions. 

\section*{Code and data availability}
The dataset considered in this work and the computer code for implementing our method in \textsf{R} are made available on \textsf{GitHub} at \url{https://github.com/STATJANG/MPGIG_INGARCH}.

\singlespacing

\end{document}